\documentclass[letterpaper,10pt,twoside]{article}
\usepackage[utf8]{inputenc}
\usepackage[english]{babel}
\usepackage[T1]{fontenc}
\usepackage{lmodern}
\usepackage{amsmath,amssymb,amsthm} 

\newcommand{\im}{\mathrm{i}}
\newcommand{\R}{\mathbb{R}}
\newcommand{\C}{\mathbb{C}}
\newcommand{\defeq}{:=}
\newcommand{\cH}{\mathcal{H}}
\newcommand{\rH}{\mathrm{H}}
\newcommand{\rL}{\mathrm{L}}
\newcommand{\cL}{\mathcal{L}}
\newcommand{\xd}{\mathrm{d}}

\newcommand{\contb}{\mathrm{C_b}}
\newcommand{\contvi}{\mathrm{C_0}}
\newcommand{\toi}{\hookrightarrow}
\newcommand{\id}{\mathrm{id}}
\newcommand{\hr}{\mathrm{H}}
\newcommand{\sr}{\mathrm{S}}
\newcommand{\hrc}{\mathrm{H,c}}
\newcommand{\src}{\mathrm{S,c}}
\newcommand{\hrp}{\mathrm{H,p}}
\newcommand{\srp}{\mathrm{S,p}}
\newcommand{\hrcp}{\mathrm{H,cp}}
\newcommand{\hacp}{\mathrm{\hat{H},cp}}
\newcommand{\srcp}{\mathrm{S,cp}}
\newcommand{\st}{\mathrm{\tilde{S}}}
\newcommand{\stc}{\mathrm{\tilde{S},c}}
\newcommand{\stcp}{\mathrm{\tilde{S},cp}}
\newcommand{\bsf}{B}
\newcommand{\bst}{\mathcal{B}}
\newcommand{\ha}{\mathrm{\hat{H}}}
\newcommand{\hac}{\mathrm{\hat{H}},c}
\newcommand{\coh}{K}
\newcommand{\cohn}{\tilde{K}}
\newcommand{\cohh}{K^{\hr}}
\newcommand{\cohs}{K^{\sr}}
\newcommand{\cohnh}{\tilde{K}^{\hr}}
\newcommand{\cohns}{\tilde{K}^{\sr}}
\newcommand{\cohnsr}{k^{\sr}}
\newcommand{\coha}{\hat{K}}
\newcommand{\cohah}{\hat{K}^{\ha}}
\newcommand{\cohas}{\hat{K}^{\sr}}
\newcommand{\bsfa}{\hat{B}}
\newcommand{\bsta}{\hat{\mathcal{B}}}
\newcommand{\lino}{\mathrm{lin}}

\theoremstyle{definition}
\newtheorem{dfn}{Definition}[section]

\theoremstyle{plain}
\newtheorem{lem}[dfn]{Lemma}
\newtheorem{prop}[dfn]{Proposition}
\newtheorem{thm}[dfn]{Theorem}

\allowdisplaybreaks

\begin{document}

\begin{titlepage}
\title{\textbf{The Schrödinger representation\\ and its relation to\\ the holomorphic representation\\ in linear and affine field theory}}
\author{Robert Oeckl\footnote{email: robert@matmor.unam.mx}\\ \\
Centro de Ciencias Matemáticas,\\
Universidad Nacional Autónoma de México,\\
Campus Morelia, C.P.~58190, Morelia, Michoacán, Mexico}
\date{UNAM-CCM-2011-2\\ 23 September 2011\\ 10 September 2012 (v2)}

\maketitle

\vspace{\stretch{1}}

\begin{abstract}
We establish a precise isomorphism between the Schrödinger representation and the holomorphic representation in linear and affine field theory. In the linear case this isomorphism is induced by a one-to-one correspondence between complex structures and Schrödinger vacua. In the affine case we obtain similar results, with the role of the vacuum now taken by a whole family of coherent states. In order to establish these results we exhibit a rigorous construction of the Schrödinger representation and use a suitable generalization of the Segal-Bargmann transform. Our construction is based on geometric quantization and applies to any real polarization and its pairing with any Kähler polarization.

\end{abstract}

\vspace{\stretch{1}}
\end{titlepage}

\section{Introduction}

The Schrödinger representation and the holomorphic representation both play important roles in quantum field theory and quantization schemes leading to such. While the former naturally combines with the Feynman path integral, the latter has better analytic properties and is more closely related to the Fock-space point of view, preferred in operator approaches. In order to be able to use the advantages of both, it is important to understand precisely the relationship between them. The primary aim of the present paper is to contribute to this understanding.

In order to relate the two types of representation, we use geometric quantization. This allows to view both as special cases of a more general framework. Roughly speaking, the Hilbert spaces of both representations arise as subspaces of the same ``prequantum'' Hilbert space. ``Pairing'' via the inner product in this larger Hilbert space then allows to extract an isomorphism between the two representations. This works reasonably well for theories with finitely many degrees of freedom \cite{Woo:geomquant} and is particularly clear for linear theories. Then, holomorphic representations are parametrized by certain complex structures while there is just one Schrödinger representation. Moreover, for any such complex structure pairing leads to an isomorphism between the associated holomorphic representation and the Schrödinger representation. This isomorphism is a version of the Segal-Bargmann transform \cite{Bar:hilbanalytic}.

In field theory, i.e., when the space of solutions of the equations of motions (or phase space) is infinite-dimensional, the situation is considerably more complicated. Again, the theories best understood are linear theories. While the holomorphic representation is not much affected by the transition to infinitely many degrees of freedom, the Schrödinger representation shows qualitatively new features. The first is that the very definition of the Hilbert space of Schrödinger wave function now depends on an additional datum, the \emph{vacuum} (see e.g.\ \cite{Jac:schroedinger}). The second is that wave functions can no longer be regarded as functions on the natural configuration space, but need to be defined as functions on a larger space. Although it is known how this latter problem can be addressed in a universal fashion \cite{BaSeZh:algconstqft}, most treatments in the literature that do take it into account resolve it in a manner dependent on the structure of a specific (class of) field theories and underlying spacetimes, see e.g., \cite{GlJa:quantumphysics}. A rigorous (and universal) definition of the Schrödinger representation is thus an important aspect of our treatment.

Our principal interest is throughout the paper in the functional nature of the Schrödinger and holomorphic representation spaces and their relation arising from geometric quantization. (We use the customary terminology ``representation'' throughout, even when no action of an algebra on the space is considered.) Since the holomorphic Hilbert space is a reproducing kernel Hilbert space, it is naturally equipped with a family of coherent states, including a vacuum state. As already mentioned, the pairing in the finite-dimensional case is realized as a kind of Segal-Bargmann transform. (This connection is made precise in Appendix~\ref{sec:bst}.) After reviewing basics of geometric quantization in Section~\ref{sec:shgquant} and introducing the relevant structures for linear field theory in Section~\ref{sec:linftstruc} we present this in a suitable coordinate-free language in Section~\ref{sec:pisom}. In particular, we obtain a corresponding family of coherent states in the Schrödinger Hilbert space, which depends on the complex structure used to define the holomorphic representation. This leads to a correspondence between compatible complex structures and certain bilinear forms characterizing the Schrödinger vacuum (Section~\ref{sec:fdcorr}).

Generalizing the Schrödinger representation to the infinite-dimensional case in Section~\ref{sec:infinis} includes a generalization of the coherent states. These can now be used to \emph{define} the isomorphism to the corresponding holomorphic Hilbert space for which we refer to the treatment in \cite{Oe:holomorphic}. This is done in Section~\ref{sec:relholom}, including the generalization to the infinite-dimensional case of the exact correspondence between compatible complex structures and bilinear forms parametrizing Schrödinger vacua. We proceed in Section~\ref{sec:linobs} to consider the action of (mainly) linear observables as prescribed by geometric quantization and show that the constructed isomorphism is indeed one of representations (with the usual functional analytic caveats).

A considerable part of the paper (Section~\ref{sec:affine}) is then devoted to generalizing all obtained results from linear to affine field theory, i.e, where the space of solutions naturally is an affine space. The main reference for the holomorphic representation is here \cite{Oe:affine}. In Section~\ref{sec:abasic}
the basic structures are defined in the affine case, Section~\ref{sec:pisoma} discusses the pairing in the finite-dimensional case. The Schrödinger representation of affine field theory is constructed in Section~\ref{sec:infinisa} and related to the holomorphic representation in Section~\ref{sec:relholoma}. A key difference to the linear theory arises from the absence of any naturally preferred vacuum state. Its role may be seen to be taken instead by the full family of coherent states. Nevertheless, a natural isomorphism between the Schrödinger and the holomorphic representation emerges, based on the same correspondence between complex structures and bilinear forms as for the linear theory. Finally, the representation of affine observables is considered in Section~\ref{sec:aobs}, leading to an isomorphism of representations.

The present paper is set up in such a way that it may alternatively be read as aimed primarily at a constructive definition of the Schrödinger representation, based solely on the holomorphic representation with its coherent states as well as the pairing, and without any representation theoretic input. For this reason we also include an explicit proof of the denseness of the coherent states in the Schrödinger representation which is detailed in Appendix~\ref{sec:proof}.

The applicability of the paper's results to any real polarization in place of the Schrödinger representation is explained in Section~\ref{sec:realpol}. Section~\ref{sec:minimizing} explains how the ingredients to define the Schrödinger representation can be ``economized'' from an intrinsic point of view. In Section~\ref{sec:CCQ} we explore the relation to previous results that have been obtained concerning the explicit correspondence between complex structures and Schrödinger vacua. Indeed, concentrating on a particular class of linear theories in certain spacetimes, the authors in \cite{CCQ:schroefock,CCQ:schroecurv} derive the Schrödinger vacuum corresponding to a given complex structure as well as aspects of the representation theory of linear observables. We show how our approach reproduces and further illuminates their results.

A brief outlook is presented in Section~\ref{sec:outlook}.

\section{Aspects of geometric quantization}
\label{sec:shgquant}

Both the Schrödinger representation and the holomorphic representation arise as special cases in the framework of geometric quantization. The latter is thus well suited to study the relation between the two. We shall recall in the following how the two representations arise in geometric quantization in the context of Lagrangian field theory. Note that the following account is imprecise and inaccurate in various respects, but plays a motivational role for the remainder of this paper. For a much more complete exposition of geometric quantization, see \cite{Woo:geomquant}. We largely follow conventions and notations as in \cite{Oe:holomorphic,Oe:affine}.

\subsection{Ingredients from Lagrangian field theory}

Suppose a classical field theory is defined on a smooth spacetime manifold $T$ of dimension $d$ and determined by a first order Lagrangian density $\Lambda(\varphi,\partial\varphi,x)$ with values in $d$-forms on $T$. Here $x\in T$ denotes a point in spacetime, $\varphi$ a field configuration at a point and $\partial\varphi$ the spacetime derivative at a point of a field configuration. We shall assume that the configurations are sections of a trivial vector bundle over $T$. We shall also assume in the following that all fields decay sufficiently rapidly at infinity where required (i.e., where regions or hypersurfaces are non-compact). Given a spacetime region $M$ and a field configuration $\phi$ in $M$ its \emph{action} is given by
\begin{equation}
 S_M(\phi)\defeq \int_M \Lambda(\phi(\cdot),\partial\phi(\cdot),\cdot) .
\end{equation}
Given a hypersurface $\Sigma$ we denote by $A_\Sigma$ the space of (germs of) solutions of the Euler-Lagrange equations in a neighborhood of $\Sigma$. The \emph{symplectic potential} is then the one-form on $A_\Sigma$ defined as\footnote{The negative sign in definition (\ref{eq:sympot}) might seem unusual. In most quantization schemes it is in fact irrelevant as the orientation of $\Sigma$ can be chosen at will. The sign is relevant, however, in quantizations targeting the general boundary formulation of quantum theory \cite{Oe:gbqft}. We put it here to ensure compatibility with \cite{Oe:affine}.}
\begin{equation}
 (\theta_{\Sigma})_{\phi}(X)\defeq -\int_\Sigma X^a \left.\partial_\mu\lrcorner\frac{\delta \Lambda}{\delta\, \partial_\mu\varphi^a}\right|_\phi .
\label{eq:sympot}
\end{equation}
Here $\phi\in A_\Sigma$ while $X$ is a tangent vector to $\phi$, i.e., an element of the space $T_\phi A_\Sigma$ of solutions linearized around $\phi$.
Also associated with a hypersurface is the \emph{symplectic form}, the two-form on $A_\Sigma$ given by the exterior derivative of the symplectic potential, 
\begin{equation}
(\omega_\Sigma)_\phi(X,Y)  \defeq(\xd\theta_\Sigma)_\phi(X,Y) .
\label{eq:sympl}
\end{equation}

In geometric quantization (as in most quantization schemes), the symplectic structure (\ref{eq:sympl}) is an essential ingredient. It serves in particular to define a correspondence principle between classical and quantum observables, relating the Poisson bracket (arising from the symplectic structure) between the former to commutators between the latter. Moreover, in the flavors of geometric quantization we will consider, also the symplectic potential (\ref{eq:sympot}) plays a role. A Hilbert space and operators on it are constructed based on the data $(A_\Sigma,\theta_\Sigma,\omega_\Sigma)$, as well as additional data to be specified. Since the mentioned data are associated to a hypersurface $\Sigma$, so will be the constructed Hilbert space and operators.
This may indeed be desirable if quantization is to target the general boundary formulation of quantum theory, where a Hilbert space is associated to each hypersurface \cite{Oe:gbqft}. On the other hand, in more conventional schemes the aim of quantization is to construct a single Hilbert space for the quantum system. The usual way to deal with the hypersurface dependence in field theory is then roughly as follows. In a globally hyperbolic spacetime all the spaces $A_\Sigma$ of solutions for spacelike Cauchy hypersurfaces $\Sigma$ can be identified with a global space of solutions $A$. Moreover, the associated symplectic forms $\omega_\Sigma$ all give rise to the same symplectic form $\omega$ on this global space. (This is not true for the symplectic potential though.)

Whatever the precise context, we shall drop hypersurface indices in the following as we will be interested in the quantization problem for a single triple $(A,\theta,\omega)$ of space of solutions, symplectic potential and symplectic form.

\subsection{Ingredients from geometric quantization}

Geometric quantization of a classical phase space $A$ with symplectic two-form $\omega$ proceeds in two steps: A hermitian line bundle $B$, the \emph{prequantum bundle} is constructed over $A$, equipped with a connection $\nabla$ whose curvature is given by the symplectic form $\omega$. The \emph{prequantized Hilbert space} $H$ is then given by square-integrable sections of $B$ with respect to a measure $\mu$ that is invariant under symplectic transformations. The inner product between sections $v',v$ is thus,
\begin{equation}
\langle v',v\rangle=\int (v'(\eta),v(\eta))_\eta\,\xd\mu(\eta),
\label{eq:pqip}
\end{equation}
where $(\cdot,\cdot)_\eta$ denotes the hermitian inner product on the fiber over $\eta\in A$. Note that a symplectic potential, i.e., a one-form $\theta$ over $A$ such that $\omega=\xd\theta$ gives rise to a trivialization of the bundle $B$ through the choice of a special section $s:A\to B$ that satisfies
\begin{equation}
 \nabla_X s=-\im\,\theta(X)\cdot s
\label{eq:trivsec}
\end{equation}
for all vector fields $X$ on $A$. Any other section of $B$ can then be obtained as $\psi s$, where $\psi$ is a complex valued function on $A$. We then have
\begin{equation}
 \nabla_X (\psi s)=(-\im\, \theta(X)\cdot\psi +\xd\psi(X))\, s .
\label{eq:trivid}
\end{equation}
 Moreover, by adjusting the overall normalization of $s$ if necessary we can arrange
\begin{equation}
 (s(\eta),s(\eta))_\eta=1 \qquad \forall \eta\in A .
\label{eq:normsec}
\end{equation}
The inner product (\ref{eq:pqip}) may then be written as,
\begin{equation}
\langle \psi' s, \psi s\rangle=\int \overline{\psi'(\eta)} \psi(\eta)\,\xd\mu(\eta) .
\label{eq:trivip}
\end{equation}

To a classical observable $F:A\to\R$ is associated its Hamiltonian vector field $X_F$ on $A$ determined by the equation
\begin{equation}
 X_F \lrcorner\, \omega=-\xd F.
\label{eq:hvectobs}
\end{equation}
Geometric quantization assigns to $F$ the operator $\check{F}:H\to H$ on the prequantum Hilbert space given by,
\begin{equation}
 \check{F} v \defeq -\im\, \nabla_{X_F} v + F v .
\label{eq:gquantobs}
\end{equation}
The quantum observables constructed in this way can then be seen to satisfy the Dirac quantization conditions. In particular, this means that if $F$ is the constant function with value $1$, then $\check{F}$ is the identity operator. Moreover, the commutator of quantum observables is related to the quantization of the Poisson bracket of the corresponding classical observables. Concretely, for two classical observables $F,G$ we get,
\begin{equation}
 \check{F}\check{G}-\check{G}\check{F}= -2\im(\omega(X_F,X_G))\check\, .
\label{eq:crobs}
\end{equation}
We also remark that if there is a symplectic potential $\theta$ with associated special section $s$ satisfying (\ref{eq:trivsec}) we can rewrite (\ref{eq:gquantobs}) for functions $\psi$ on $A$ as,
\begin{equation}
 \check{F} (\psi s) = (-\theta(X_F)\cdot \psi -\im\, \xd\psi(X_F) + F \psi) s .
\label{eq:gquantobstriv}
\end{equation}

While the prequantum Hilbert space $H$ represents only an intermediate construction, the actual Hilbert space of states $\cH$ is obtained by a suitable restriction of $H$ trough a \emph{polarization}. This is the second step. A polarization consists roughly of a choice of Lagrangian subspace $P_\eta$ of the complexified tangent space $T_\eta A^\C$ for each point $\eta\in A$. One then defines \emph{polarized sections} of $B$ to be those $v:A\to B$ satisfying
\begin{equation}
\nabla_{\overline{X}} v=0,
\label{eq:poleq}
\end{equation}
where $X$ is a complex vector field valued at each point $\eta\in A$ in the polarized subspace $P_\eta\subseteq T_\eta A^\C$. Here $\overline{X}$ denotes the complex conjugation of $X$. A possibly complex one-form $\theta$ on $A$ such that $\xd\theta=\omega$ and such that
\begin{equation}
\theta(\overline{X})=0
\label{eq:adaspot}
\end{equation}
for all complex vector fields $X$ of this type is called a symplectic potential \emph{adapted to the polarization}. Given such an adapted symplectic potential and an associated section $s:A\to B$ satisfying (\ref{eq:trivsec}) the condition (\ref{eq:poleq}) can be rewritten for sections $v=\psi s$ as a condition on the admissible functions $\psi:A\to\C$,
\begin{equation}
\xd \psi(\overline{X})=0 .
\label{eq:polfunc}
\end{equation}
Note that in the case where $\theta$ is complex $s$ cannot in general be chosen to satisfy the normalization condition (\ref{eq:normsec}) as well.

The restriction of $H$ to the polarized sections yields the Hilbert space $\cH$. An immediate problem that arises is that not all quantum observables defined via (\ref{eq:gquantobs}) leave the subspace $\cH\subset H$ invariant. Addressing this then requires further refinements of the geometric quantization scheme. However, for the purposes of the present paper we may ignore this problem since we shall be interested only in a limited class of observables which do leave $\cH$ invariant.

\subsection{Polarizations and pairings}
\label{sec:polpa}

We shall be interested in two types of polarizations: Real polarizations and Kähler polarizations. In the case of a real polarization the Lagrangian subspaces $P_\eta$ of the complexified tangent spaces $T_\eta A^\C$ for $\eta\in A$ arise simply as complexifications of Lagrangian subspaces of the real tangent spaces $T_\eta A$. The vector fields $X$ valued at each point in the polarized subspace and appearing in equations (\ref{eq:poleq}) and (\ref{eq:adaspot}) can then be taken to be real without loss of generality. Moreover, an adapted symplectic potential will be real. The real polarization of particular interest in the present paper is the Schrödinger polarization. This polarization is defined for a point $\eta\in A$ by the subspace $M_\eta\subset T_\eta A$, which is spanned by the ``momentum'' directions. In a field theory context this means roughly the directions spanned by the derivative of the field perpendicular to the hypersurface. The symplectic potential adapted to the Schrödinger polarization is precisely the symplectic potential exhibited in formula (\ref{eq:sympot}). As is easily seen from the formula, this symplectic potential depends indeed only on the vector directions $X^a$ representing field values and not on derivatives, thus satisfying condition (\ref{eq:adaspot}).

In the holomorphic case, the polarization is induced by a complex structure $J_\eta$ on each tangent space $T_\eta A$, which is at the same time a symplectic transformation. Then, $\phi\mapsto \frac{1}{2}(\phi-\im J_\eta \phi)$ projects onto the polarized subspace $P_\eta\subseteq T_\eta A^\C$. At least locally, there exists then a \emph{Kähler potential} $K:A\to\R$ and an adapted complex symplectic potential $\Theta$ such that
\begin{equation}
 \Theta=-\im\sum_i \frac{\partial K}{\partial z_i}\xd z_i,
\label{eq:ksympot}
\end{equation}
where $\{z_i\}$ are local holomorphic coordinates with respect to the complex structure $J$. We can choose a (local) section $u$ of $B$ satisfying (\ref{eq:trivsec}) with respect to the complex one-form $\Theta$ to trivialize $B$. Then, general sections of $B$ can be obtained as $\psi u$ with $\psi$ a complex valued function on $A$. Moreover, equation (\ref{eq:polfunc}) translates to the condition that $\psi$ is a \emph{holomorphic} function on $A$. Since $\Theta$ is complex, the section $u$ cannot be normalized in analogy to (\ref{eq:normsec}). However, it can be related to the section $s$ that satisfies (\ref{eq:trivsec}) with respect to a given real symplectic potential $\theta$ as well as (\ref{eq:normsec}). Indeed, let $\alpha$ be the complex function on $A$ such that $u=\alpha s$. Then, we can use (\ref{eq:trivip}) to write the inner product on $\cH$ as follows,
\begin{equation}
 \langle \psi' u, \psi u\rangle=\int \overline{\psi'(\eta)} \psi(\eta)\, |\alpha(\eta)|^2\,\xd\mu(\eta) .
\label{eq:holip}
\end{equation}

Given two different polarizations, the fact that both polarized Hilbert spaces $\cH$, $\cH'$ arise as subspaces of the same prequantum Hilbert space $H$ gives us a means to ``compare'' them, by taking the inner product between a section in $\cH$ and a section in $\cH'$. If this gives rise to a non-degenerate pairing, we can use it to construct an isomorphism between $\cH$ and $\cH'$. This is precisely the device that we shall use to relate the Schrödinger and the holomorphic representations. The former is given by the Schrödinger polarization and the latter by a Kähler polarization. Concretely, suppose that $\theta$ is the symplectic potential adapted to the Schrödinger polarization and $\Theta$ is a symplectic potential adapted to a Kähler polarization. We define sections $s$ and $u$ of the prequantum bundle and a function $\alpha$ as explained above. Then for $\psi,\psi'$ complex valued functions on $A$ we get the pairing
\begin{equation}
 \langle \psi' s, \psi u \rangle=\int \overline{\psi'(\eta)} \psi(\eta) \alpha(\eta)\,\xd\mu(\eta) .
\label{eq:pairing}
\end{equation}

\section{Linear field theory}

\subsection{Basic structures}
\label{sec:linftstruc}

We specialize in this Section to the first type of field theory of principal interest in the present paper: \emph{linear field theory}. At the same time we make precise in this context some of the corresponding notions introduced in Section~\ref{sec:shgquant}.

Firstly, we suppose that the space of solutions, which we shall denote in the linear case by $L$ rather than by $A$, is a real vector space. This allows to identify canonically all the tangent spaces $T_\xi L$ with $L$ itself. The symplectic potential may then be seen as a map $[\cdot,\cdot]:L\times L\to\R$, linear in the second argument, where we use the notation $[\xi,\tau]\defeq \theta_\xi(\tau)$. Our second key assumption is that the symplectic potential is equivariant with respect to vector addition, i.e., $[\cdot,\cdot]$ is linear also in the first argument. This implies in turn that the symplectic form is independent of the base point and may be viewed as an anti-symmetric bilinear map $\omega:L\times L\to\R$ given in terms of the symplectic potential as follows,
\begin{equation}
 \omega(\xi,\xi')=\frac{1}{2}[\xi,\xi']-\frac{1}{2}[\xi',\xi] \qquad\forall\xi,\xi'\in L .
\label{eq:sptosf}
\end{equation}
We shall assume moreover that the symplectic form is non-degenerate.

To consider the Schrödinger representation it is convenient to define the subspaces $M$ and $N$ of $L$ as follows:
\begin{equation}
 M\defeq\{\tau\in L: [\xi,\tau]=0\; \forall \xi\in L\}\qquad
 N\defeq\{\tau\in L: [\tau,\xi]=0\; \forall \xi\in L\} .
\label{eq:defmn}
\end{equation}
It follows directly from the definitions that $M$ and $N$ are isotropic subspaces of $L$. Also, using the non-degeneracy of the symplectic form, it follows that $M\cap N=\{0\}$. We shall make the additional assumption that $M$ and $N$ together generate $L$ as a vector space, so we even have $L=M\oplus N$. It follows then, that $M$ and $N$ are also coisotropic and thus Lagrangian subspaces of $L$.
Supposing that $[\cdot,\cdot]$ is the symplectic potential adapted to the Schrödinger polarization, i.e.\ given by (\ref{eq:sympot}) in field theory, then $M$ is precisely the real subspace of $L$ defining the Schrödinger polarization. Recall that there is a special section $s$ of the prequantum bundle $B$ satisfying (\ref{eq:trivsec}) with respect to $[\cdot,\cdot]$ as well as (\ref{eq:normsec}). From here onwards we fix this choice of $s$. The sections defining the Schrödinger polarized Hilbert space then take the form $\psi s$, where $\psi$ is a complex function on the quotient space of ``field configurations on the hypersurface''
\begin{equation}
 Q\defeq L/M.
\label{eq:defq}
\end{equation}
We denote the quotient map $L\to Q$ by $q$. We also remark that due to the definition of $M$, the symplectic potential $[\cdot,\cdot]$ may be viewed alternatively as a map $L\times Q\to\R$. We shall occasionally make use of this fact without making the distinction explicit in the notation. Furthermore, restricting $[\cdot,\cdot]$ to a map $M\times Q\to\R$ makes it non-degenerate. (To see this, identify $Q$ and $N$, note that $[\cdot,\cdot]$ coincides with $2\omega$ on $M\times N$, and use that $M$ and $N$ are both coisotropic subspaces of $L$.)

For the holomorphic representation, i.e., a Kähler polarization, we need as an additional ingredient apart from the classical data already described a \emph{complex structure} on the tangent spaces of $L$. Since these tangent spaces are all canonically identified with $L$ itself and the symplectic form is independent of the base point it will suffice to consider a single complex structure on $L$. Thus, the complex structure is a linear map $J:L\to L$ satisfying $J^2=-\id$ and $\omega(J(\cdot), J(\cdot))=\omega(\cdot,\cdot)$. This gives rise to the symmetric bilinear form $g:L\times L\to\R$ by
\begin{equation}
 g(\tau,\xi)\defeq 2\omega(\tau,J \xi) \qquad\forall\tau,\xi\in L .
\label{eq:kmetric}
\end{equation}
We shall assume that this form is positive definite and makes $L$ into a real separable Hilbert space. It is then true that the sesquilinear form
\begin{equation}
\{\tau,\xi\}\defeq g(\tau,\xi)+2\im\omega(\tau,\xi) \qquad\forall\tau,\xi\in L
\label{eq:kcip}
\end{equation}
makes $L$ into a complex separable Hilbert space, where multiplication with $\im$ is given by applying $J$. Note that by construction $J$ is continuous in the topology defined by (\ref{eq:kmetric}) or equivalently by (\ref{eq:kcip}). Moreover, by combining $J$ with the Riesz representation theorem, the continuous real-linear maps $L\to\R$ are in one-to-one correspondence with elements $\xi\in L$ via
\begin{equation}
 \tau\mapsto \omega(\xi,\tau) .
\end{equation}
We also require that the symplectic potential $[\cdot,\cdot]$ is continuous in the topology defined on $L$.

As discussed above, the complex structure $J$ defines a polarization and implies the existence of a Kähler potential. The Kähler potential is not unique, but a natural choice is given by,
\begin{equation}
 K(\xi)\defeq\frac{1}{2}g(\xi,\xi) .
\end{equation}
The adapted symplectic potential $\Theta:L\times L\to\C$ from (\ref{eq:ksympot}) is then,
\begin{equation}
 \Theta(\tau,\xi)=-\frac{\im}{2}\{\tau,\xi\} .
\label{eq:kpotlin}
\end{equation}
Define the complex function $\alpha$ on $L$ by
\begin{equation}
\alpha(\xi)\defeq\exp\left(\frac{\im}{2} [\xi,\xi]-\frac{1}{4} g(\xi,\xi)\right) .
\label{eq:alphagq}
\end{equation}
As is easily verified this satisfies
\begin{equation}
 \xd\alpha=-\im\alpha(\Theta-\theta) .
\end{equation}
Thus, it follows from (\ref{eq:trivid}) that the section $u\defeq \alpha s$ of $B$ satisfies (\ref{eq:trivsec}) with respect to $\Theta$.

In order to compare the two types of polarizations it is useful to introduce a few further structures. Since $M$ is a Lagrangian subspace of $L$, $M\oplus J M$ is an orthogonal decomposition of $L$ as a real Hilbert space. In particular, $J M$ is another complement of $M$ in $L$, which does not necessarily coincide with $N$, compare (\ref{eq:defmn}). Moreover, $J M$ is also Lagrangian subspace of $L$. We equip the quotient space $Q$ defined in (\ref{eq:defq}) with the quotient norm, making it into a real Hilbert space. Let $j$ be the unique linear map $Q\to L$ such that $q\circ j=\id_Q$ and such that $j(Q)\subseteq J M$. Then, $j$ is in fact an isometric isomorphism of real Hilbert spaces and $j(Q)=J M$.

\subsection{Pairing and isomorphism in finite dimensions}
\label{sec:pisom}

In the case where the space of solutions is linear and finite-dimensional 
it is known \cite{Woo:geomquant} how to implement the remaining ingredients of the geometric quantization program, notably the inner product formulas (\ref{eq:trivip}) and (\ref{eq:holip}) as well as the pairing (\ref{eq:pairing}). We shall consider this here from a point of view where the holomorphic representation is the primary object and properties of the Schrödinger representation are derived. We shall follow the setup of the holomorphic representation as presented in \cite{Oe:holomorphic}.

For the Schrödinger representation, the polarized ``wave functions'' $\psi$ multiplying the section $s$ depend on the quotient space $Q$ rather than the full space of solutions $L$. Thus, in order for the inner product (\ref{eq:trivip}) to be normalizable the integral has to be performed over $Q$ rather than over $L$. Then, $\mu$ is a Lebesgue measure on $Q$ which we shall denote by $\mu_Q$. It will turn out to be convenient to normalize $\mu_Q$ such that 
\begin{equation}
\int \exp\left(-g(j(\phi),j(\phi))\right)\,\xd\mu_Q(\phi) =1 .
\label{eq:muq}
\end{equation}
In the following, we shall denote by $\cH^{\sr}$ the complex Hilbert space $\rL^2(Q,\mu_Q)$ of complex square-integrable functions on $Q$ with respect to the measure $\mu_Q$. We denote its inner product by $\langle\cdot,\cdot\rangle^{\sr}$.

For the holomorphic representation, the inner product (\ref{eq:holip}) does make immediate sense with $\mu$ a Lebesgue measure that we shall denote by $\mu_L$. To see this, note the factor $|\alpha(\xi)|^2$ appearing in the integral is a Gaussian given by
\begin{equation}
|\alpha(\xi)|^2=\exp\left(-\frac{1}{2} g(\xi,\xi)\right) .
\end{equation}
We choose the normalization of $\mu_L$ such that
\begin{equation}
\int |\alpha(\xi)|^2\,\xd\mu_L(\xi) =1 .
\end{equation}
It is convenient to define a new measure $\nu_L$ on $L$ given by
\begin{equation}
 \xd\nu_L\defeq \exp\left(-\frac{1}{2} g(\xi,\xi)\right)\xd\mu_L .
\label{eq:fdhmeasure}
\end{equation}
In the following we shall denote by $\cH^{\hr}$ the complex Hilbert space $\rH^2(L,\nu_L)$ of holomorphic square-integrable functions on $L$ with respect to the measure $\nu_L$. We denote its inner product by $\langle\cdot,\cdot\rangle^{\hr}$.

Following Section~\ref{sec:polpa} the idea is now to relate the two Hilbert spaces $\cH^{\sr}$ and $\cH^{\hr}$ using the pairing (\ref{eq:pairing}) and extract from it an isomorphism. However, it is not quite straightforward to define the integral in (\ref{eq:pairing}) for general elements of $\cH^{\sr}$ and $\cH^{\hr}$. To circumvent this difficulty we restrict to a dense subspace where the definition of the integral presents no problem. It turns out to be convenient to choose the subspace spanned by the coherent states to this end. Following the conventions in \cite{Oe:holomorphic}, we recall that the standard coherent state $\coh_\tau$ associated to a solution $\tau\in L$ is represented by the wave function $\cohh_\tau$ in $\cH^{\hr}$ given as,
\begin{equation}
 \cohh_{\tau}(\xi)=\exp\left(\frac{1}{2} \{\tau,\xi\}\right) .
\label{eq:hcoh}
\end{equation}
Crucially, these coherent states satisfy the reproducing property
\begin{equation}
\langle\cohh_{\tau},\psi\rangle^{\hr}=\psi(\tau)
\label{eq:cohrp}
\end{equation}
for all $\psi\in\cH^{\hr}$. At a later point, it will be convenient to consider also the normalized versions $\cohn_\tau$ of the coherent states with wave functions given by
\begin{equation}
 \cohnh_{\tau}(\xi)=\exp\left(\frac{1}{2} \{\tau,\xi\} - \frac{1}{4} g(\tau,\tau)\right) .
\label{eq:hncoh}
\end{equation}
We recall the inner product between coherent states,
\begin{equation}
 \langle \cohh_\tau,\cohh_{\tau'}\rangle^{\hr}=\exp\left(\frac{1}{2}\{\tau',\tau\}\right) .
\label{eq:ipcoh}
\end{equation}
We shall denote the subspace of $\cH^{\hr}$ spanned by coherent states by $\cH^{\hrc}$. Recall that $\cH^{\hrc}$ is dense in $\cH^{\hr}$. (In \cite{Oe:holomorphic} this is Proposition~3.12.)

In order to describe the isomorphism between the Hilbert spaces of the different representations it is convenient to introduce the continuous function $\bsf:L\times Q\to\C$ given as follows,
\begin{multline}
 \bsf(\xi,\phi)\defeq \exp\left(\{j(\phi),\xi\}-\frac{\im}{2} [j(\phi),j(\phi)]-\frac{1}{2} g(j(\phi),j(\phi)) \right. \\
\left.
+\frac{1}{4} g(\xi,\xi) -\frac{1}{2} \{j\circ q(\xi),\xi\}\right) .
\label{eq:sbt}
\end{multline}
Note that $\bsf$ is holomorphic in its first argument. The following Proposition provides a precise implementation of the pairing (\ref{eq:pairing}) and a description of the induced isomorphism. Compare in particular the right hand sides of (\ref{eq:ppropf}) and (\ref{eq:ppropfinv}) with (\ref{eq:pairing}).

\begin{prop}
\label{prop:bstfc}
There is a subspace $\cH^{\src}\subseteq \cH^{\sr}$ and an isometric isomorphism $\bst:\cH^{\src}\to \cH^{\hrc}$ with the following properties:
\begin{align}
 (\bst\psi)(\xi) & = \int  \psi(\phi) \bsf(\xi,\phi)\,\xd\mu_Q(\phi) && \forall \psi\in\cH^{\src} , \label{eq:deff} \\
 (\bst^{-1}\psi)(\phi) & = \int \psi(\xi) \overline{\bsf(\xi,\phi)}\,\xd\nu_L(\xi) &&
\forall \psi\in \cH^{\hrc}, \label{eq:deffinv} \\
 \langle \psi', \bst\psi \rangle^{\hr} & = \int \overline{\psi'(\xi)}\, \psi(q(\xi))\, \overline{\alpha(\xi)} \,\xd\mu_L(\xi) && \forall \psi'\in \cH^{\hrc}\, \forall \psi\in\cH^{\src}, \label{eq:ppropf} \\
 \langle \psi', \bst^{-1}\psi \rangle^{\sr} & = \int \overline{\psi'(q(\xi))}\, \psi(\xi)\, \alpha(\xi) \,\xd\mu_L(\xi) && \forall \psi'\in \cH^{\src}\, \forall \psi\in\cH^{\hrc} . \label{eq:ppropfinv}
\end{align}
Also, the holomorphic coherent state wave function $\cohh_\tau$ associated with $\tau\in L$ is mapped under $\bst^{-1}$ to the following wave function:
\begin{equation}
 \cohs_{\tau}(\phi)\defeq
 (\bst^{-1} \cohh_{\tau})(\phi)  = \overline{B(\tau,\phi)} .
\label{eq:defcohs}
\end{equation}
\end{prop}
\begin{proof}
We shall define $\bst^{-1}$ by the integral (\ref{eq:deffinv}). Setting $\psi$ to a coherent state clearly makes the integrand integrable and explicit calculation yields (\ref{eq:defcohs}). Alternatively, we may view the integral as the (slightly extended) inner product in $\cH^{\hr}$. Taking the complex conjugate of the reproducing property (\ref{eq:cohrp}) yields then immediately (\ref{eq:defcohs}). We define the subspace $\cH^{\src}\subseteq\cH^{\sr}$ to be that of linear combinations of wave functions $\cohs_\tau$ for $\tau\in L$. One may then verify (\ref{eq:deff}) for coherent states. We leave the explicit calculation to the reader. The comparison of the inner products in $\cH^{\src}$ and $\cH^{\hrc}$,i.e.,
\begin{equation}
 \langle \bst \cohs_{\tau}, \bst \cohs_{\tau'}\rangle^{\hr}
=\langle \cohs_{\tau}, \cohs_{\tau'}\rangle^{\sr}
\end{equation}
is also straightforward. Finally, it is sufficient to check (\ref{eq:ppropf}) and (\ref{eq:ppropfinv}) for coherent states.
\end{proof}

We thus have also derived the explicit form (\ref{eq:defcohs}) of the wave function of the standard coherent states in the Schrödinger representation. In particular, the wave function of the vacuum state $\coh_0$ is,
\begin{equation}
 \cohs_{0}(\phi)=\\ \exp\left(-\frac{1}{2} g(j(\phi),j(\phi)) +\frac{\im}{2} [j(\phi),j(\phi)] \right) .
\label{eq:swvac}
\end{equation}
As is customary in the Schrödinger representation, we may write wave functions $\psi\in\cH^{\sr}$ in a factorized form,
\begin{equation}
 \psi(\phi)=\widetilde{\psi}(\phi) \cohs_0(\phi) .
\label{eq:factswv}
\end{equation}
Defining on $Q$ the probability measure $\nu_Q$ via
\begin{equation}
\xd\nu_Q(\phi)\defeq \exp(-g(j(\phi),j(\phi)))\,\xd\mu_Q(\phi),
\label{eq:nuq}
\end{equation}
we have $\psi\in\cH^{\sr}$ if and only if $\widetilde{\psi}\in \rL^2(Q,\nu_Q)$. The inner product in $\cH^{\sr}$ may then be expressed as
\begin{equation}
\langle \psi', \psi\rangle^{\sr}= \int \overline{\widetilde{\psi}'(\phi)} \widetilde{\psi}(\phi)\, \xd\nu_Q(\phi) .
\end{equation}
It is now convenient to switch to the normalized coherent states. For these we shall write the factorization (\ref{eq:factswv}) as
\begin{align}
 \cohns_\tau(\phi) & =\cohnsr_\tau(\phi) \cohs_0(\phi),\quad\text{with}, \\
 \cohnsr_\tau(\phi) & =\exp\left(\{\tau,j(\phi)\}-\frac{1}{2} \{\tau,j\circ q(\tau)\}\right) . \label{eq:swvcoh}
\end{align}

Also in the Schrödinger representation, the space spanned by coherent states is dense in the Hilbert space of all states:
\begin{prop}
\label{prop:scohdensefd}
The subspace $\cH^{\src}$ is dense in $\cH^{\sr}$.
\end{prop}
While this fact is well known through other approaches, it is not obvious from the present perspective. Since our aim is to provide a rigorous and reasonably self-contained treatment, a proof is given in Appendix~\ref{sec:proof}.
The statement implies that the isometric isomorphism of pre-Hilbert spaces $\bst:\cH^{\src}\to \cH^{\hrc}$ extends to an isometric isomorphism of Hilbert spaces $\cH^{\sr}\to \cH^{\hr}$ which we shall continue to denote by $\bst$. In fact, the integral representation of $\bst$ given by expression (\ref{eq:deff}) extends to the whole Hilbert space $\cH^{\sr}$. On the other hand, the integral representation (\ref{eq:deffinv}) for the inverse transformation does not immediately extend to the whole Hilbert space $\cH^{\hr}$. This is not surprising, since the elements of $\cH^{\sr}$ are really equivalence classes of functions on $Q$ that do not actually have a well defined value at a point $\phi\in Q$. Coherent states belong to those special states that can be represented by a continuous function. It is this preferred representation that is computed by (\ref{eq:deffinv}).

The transform $\bst$ is a coordinate free version of the Segal-Bargmann transform \cite{Bar:hilbanalytic}. This is explained in detail in Appendix~\ref{sec:bst}.

\subsection{Correspondence in finite dimensions}
\label{sec:fdcorr}

The coherent states and in particular the vacuum state may be viewed as objects intrinsic to the Schrödinger representation rather than as induced from the pairing with another representation. To this end we may eliminate the explicit appearance of the complex structure in (\ref{eq:defcohs}) in favor of a more natural structure. Define the symmetric bilinear form $\Omega:Q\times Q\to\C$ as follows,
\begin{equation}
\Omega(\phi,\phi')\defeq g(j(\phi),j(\phi')) -\im [j(\phi),\phi'] .
\label{eq:jtoomega}
\end{equation}
We note that the real part of $\Omega$ is precisely the inner product on $Q$. The measures $\mu_Q$ as well as $\nu_Q$ may now be defined in terms of (the real part of) $\Omega$, without direct reference to $J$ or $g$. The same is true for the wave functions of the vacuum and indeed all coherent states. Recalling (\ref{eq:swvac}) and (\ref{eq:swvcoh}) we get
\begin{align}
 \cohs_0(\phi) & =\exp\left(-\frac{1}{2}\Omega(\phi,\phi)\right) ,
 \label{eq:swvaci} \\
 \cohnsr_\tau(\phi) & =\exp\left(\Omega(q(\tau),\phi)+\im [\tau,\phi]
  -\frac{1}{2}\Omega(q(\tau),q(\tau))-\frac{\im}{2}[\tau,\tau]\right) .
 \label{eq:swcohi}
\end{align}

We have thus arrived at the usual characterization \cite{Jac:schroedinger} of the vacuum wave function (\ref{eq:swvaci}) in the Schrödinger representation in terms of a symmetric bilinear form $\Omega:Q\times Q\to \C$ with positive definite real part. The obvious questions are now whether all such forms arise from complex structures in the way described above and whether the complex structure corresponding to such a structure is unique. The answer to both questions is affirmative, i.e., there is a one-to-one correspondence between the two structures. While the expression of $\Omega$ in terms of $J$ is given by formula (\ref{eq:jtoomega}), we shall now be interested in obtaining $J$ from $\Omega$.

Thus, suppose we are given a symmetric bilinear form $\Omega$ with positive definite real part. As a first step we define the subspace $X\subseteq L$ given by
\begin{equation}
 X\defeq\{\xi\in L: \Im\Omega(q(\xi),\phi)+[\xi,\phi]=0\;\forall \phi\in Q\} .
\label{eq:defxio}
\end{equation}
We first note that $X\cap M=\{0\}$. Indeed, suppose $\xi\in X\cap M$. Then $q(\xi)=0$ and consequently $\Im\Omega(q(\xi),\phi)=0$ for all $\phi\in Q$. Thus, $[\xi,\phi]= 0$ for all $\phi\in Q$. But the non-degeneracy of $[\cdot,\cdot]$ as a map $M\times Q\to\R$ implies then $\xi=0$.
On the other hand, the same non-degeneracy of $[\cdot,\cdot]$ implies that any linear map $Q\to\R$ may be written as $\phi\mapsto [\xi,\phi]$ for some $\xi\in M$. In particular, for any $\xi\in L$ there is an element $\gamma(\xi)\in M$ such that
\begin{equation}
\Im\Omega(q(\xi),\phi)+[\gamma(\xi),\phi]=0\quad\forall \phi\in Q .
\end{equation}
Now let $\xi\in N$. The above shows that $\xi+\gamma(\xi)\in X$. Thus, we can write any element of $N$ as a linear combination of an element of $X$ and an element of $M$. But $N$ together with $M$ generate $L$ so $X$ together with $M$ also generate $L$. Therefore $X$ is a complement of $M$ in $L$ and $L=M\oplus X$.

Let $\xi,\tau\in X$. Then,
\begin{equation}
\omega(\xi,\tau)=\frac{1}{2}[\xi,\tau]-\frac{1}{2}[\tau,\xi]
 = - \frac{1}{2}\Im\Omega(q(\xi),q(\tau)) + \frac{1}{2}\Im\Omega(q(\tau),q(\xi)) = 0 .
\end{equation}
So $X$ is isotropic in $L$. On the other hand suppose there is $\tau\in L$ such that $\omega(\xi,\tau)=0$ for all $\xi\in X$. But, for all $\xi\in X$,
\begin{equation}
0=\omega(\xi,\tau)=\frac{1}{2}[\xi,\tau]-\frac{1}{2}[\tau,\xi]
 = - \frac{1}{2}\Im\Omega(q(\xi),q(\tau)) -\frac{1}{2}[\tau,\xi] .
\end{equation}
Since $X$ is a complement of $M$, $q$ restricted to $X$ is a vector space isomorphism $X\to Q$, and so the above implies $\tau\in X$. Hence, $X$ is also coisotropic and thus Lagrangian.

By a similar reasoning as above based on the non-degeneracy of $[\cdot,\cdot]$ there is a linear map $\beta:X\to M$ such that
\begin{equation}
\Re\Omega(q(\xi),\phi)+[\beta(\xi),\phi]=0\quad\forall \xi\in X, \forall \phi\in Q .
\label{eq:defbeta}
\end{equation}
Moreover, since $\Re\Omega$ is non-degenerate, and since $q$ restricted to $X$ is an isomorphism, $\beta$ is invertible. We claim that $J\defeq\beta\oplus (-\beta^{-1}):L\to L$ is a complex structure.
To see this we decompose elements $\xi\in L$ as $\xi=\xi^M+\xi^X$, where $\xi^X\in X$ and $\xi^M\in M$. Then, for all $\xi\in L$,
\begin{equation}
J(J(\xi))=J(\beta(\xi^X)-\beta^{-1}(\xi^M))=-\beta^{-1}(\beta(\xi^X))-\beta(\beta^{-1}(\xi^M))=-\xi .
\end{equation}
As for the other key property of a complex structure, let $\xi,\tau\in L$ and observe
\begin{align}
 \omega(J\xi,J\tau)
 & = \omega(\beta(\xi^X)-\beta^{-1}(\xi^M), \beta(\tau^X)-\beta^{-1}(\tau^M)) \\
 & = -\omega(\beta(\xi^X),\beta^{-1}(\tau^M))
   -\omega(\beta^{-1}(\xi^M),\beta(\tau^X)) \\
 & = -\frac{1}{2} [\beta(\xi^X),\beta^{-1}(\tau^M)]
     + \frac{1}{2} [\beta(\tau^X),\beta^{-1}(\xi^M)] \\
 & = \frac{1}{2}\Re\Omega(q(\xi^X),q(\beta^{-1}(\tau^M)))
     - \frac{1}{2}\Re\Omega(q(\tau^X),q(\beta^{-1}(\xi^M))) \\
 & = -\frac{1}{2} [\tau^M,\xi^X] + \frac{1}{2} [\xi^M,\tau^X] \\
 & = \omega(\xi^X,\tau^M) + \omega(\xi^M,\tau^X) \\
 & = \omega(\xi,\tau) .
\end{align}
It remains to verify that the bilinear form $g:L\times L\to\R$ defined by $J$ according to (\ref{eq:kmetric}) is positive definite. Let $\xi\in L$. Then,
\begin{align}
 g(\xi,\xi)
 & = 2 \omega(\xi,J\xi) \label{eq:ipomega1} \\
 & = 2 \omega(\xi^X,J\xi^X)+2\omega(\xi^M,J\xi^M) \\
 & = -[J\xi^X,\xi^X] + [\xi^M, J\xi^M] \\
 & = -[\beta(\xi^X),\xi^X] -[\xi^M,\beta^{-1}(\xi^M)] \\
 & = \Re\Omega(q(\xi^X),q(\xi^X))
  +  \Re\Omega(q(\beta^{-1}(\xi^M)),q(\beta^{-1}(\xi^M))) \label{eq:ipomega2}
\end{align}
Positive definiteness of $\Re\Omega$ then implies positive definiteness of $g$.

We proceed to verify that the above constructed maps from complex structures $J$ to forms $\Omega$ and from forms $\Omega$ to complex structures $J$ are mutually inverse. Suppose we are given a complex structure $J$ and define $\Omega$ according to equation (\ref{eq:jtoomega}). Define now $X$ as in (\ref{eq:defxio}). Suppose that $\xi\in J M$. Then, for all $\phi\in Q$,
\begin{equation}
\Im\Omega(q(\xi),\phi)=-[j(q(\xi))),\phi]=-[\xi,\phi] .
\end{equation}
Thus $\xi\in X$. That is, $J M\subseteq X$. But since $J M$ and $X$ are both complements of $M$ in $L$ we must have $X=J M$. To see that $\beta\oplus (-\beta^{-1})$ coincides with $J$ it is sufficient to show this on the subspace $J M\subseteq L$. Thus, let $\xi\in J M$. Then, for all $\phi\in Q$ we have
\begin{multline}
 [\beta(\xi),\phi]=-\Re\Omega(q(\xi),\phi)=-g(j\circ q(\xi),j(\phi))
 =-g(\xi,j(\phi))\\ =2\omega(J\xi,j(\phi))=[J\xi,\phi] .
\end{multline}
Non-degeneracy of $[\cdot,\cdot]$ implies then $\beta=J|_{J M}$.

Conversely, suppose we are given a symmetric bilinear form $\Omega:Q\times Q\to\C$ with positive definite real part. We work out $J$ as described above and verify that formula (\ref{eq:jtoomega}) really recovers $\Omega$. For the real part observe, for all $\phi,\phi'\in Q$,
\begin{equation}
 g(j(\phi),j(\phi'))=2\omega(j(\phi),\beta\circ j(\phi'))
 =-[\beta\circ j(\phi'),\phi]=\Re\Omega(\phi',\phi) .
\end{equation}
For the imaginary part observe, for all $\phi,\phi'\in Q$,
\begin{equation}
 -[j(\phi),\phi']=\Im\Omega(q\circ j(\phi),\phi')=\Im\Omega(\phi,\phi') .
\end{equation}
We have thus proven the following proposition.
\begin{prop}
\label{prop:fdcorr}
There is a one-to-one correspondence between complex structures $J:L\to L$ and symmetric bilinear forms $\Omega:Q\times Q\to\C$ with positive definite real part.
\end{prop}

There is another concise way to express the correspondence between complex structures and bilinear forms. Consider the projector $P^+:L^{\C}\to L^{\C}$ onto the polarized subspace of the complexification $L^{\C}$ of $L$. Recall from Section~\ref{sec:polpa} that $P^{+}(\xi)=\frac{1}{2}(\xi-\im J\xi)$. Extending the maps $[\cdot,\cdot]$, $\Omega$ and $q$ to complexified spaces, we have
\begin{equation}
 \Omega(q(P^+(\xi)),\phi)=-\im [P^+(\xi),\phi]\qquad\forall\xi\in L,\forall \phi\in Q .
\end{equation}
This equation is indeed sufficient to uniquely determine the complex structure $J$ from the bilinear form $\Omega$ or vice versa. Note in particular, that $q\circ P^+$ restricted to $L$ yields an isomorphism of real vector spaces between $L$ and $Q^\C$.

We summarize the picture in the finite-dimensional case as follows. For each complex structure $J$ there is a distinct Kähler polarization leading to a distinct Hilbert space $\cH^\hr$ of polarized holomorphic wave functions. On the other hand, there is a unique Schrödinger polarization yielding a unique Hilbert space $\cH^\sr$. Through the pairing there is, for each complex structure $J$, an induced isometric isomorphism $\bst:\cH^\sr\to\cH^\hr$ of Hilbert spaces. Moreover, for each complex structure $J$, the coherent states of $\cH^\hr$ are mapped under $\bst^{-1}$ to a distinct family of coherent states in $\cH^\sr$. In particular, for each distinct $J$, the vacuum in $\cH^\hr$ is mapped to a state in $\cH^\sr$ that is characterized by a distinct symmetric bilinear form $\Omega:Q\times Q\to\C$ with positive definite real part. Moreover, any such $\Omega$ comes from some $J$.

\subsection{Schrödinger representation in infinite dimensions}
\label{sec:infinis}

After having derived in the finite-dimensional case the key elements of the Schrödinger representation we proceed here to generalize them to the case where the space $L$ of solutions is infinite-dimensional. For the case of the holomorphic representation a detailed account of this was given in \cite{Oe:holomorphic}, including an adapted exposition of the basic mathematical ingredients. We shall use much of the same ingredients here and thus refer the reader to that paper for more details in this respect.

Specifically, we suppose we are given the structures introduced in Section~\ref{sec:linftstruc}, except for a complex structure and structures derived from it. Moreover, we suppose we are given a symmetric bilinear form $\Omega:Q\times Q\to\C$ with positive definite real part. In contrast to the finite-dimensional case, we also need to add compatibility conditions concerning the topologies generated by the different algebraic structures. Specifically we require $Q$ to be complete and separable, i.e., to be a real separable Hilbert space with respect to the inner product given by the real part of $\Omega$. Also, we shall assume the imaginary part of $\Omega$ to be continuous with respect to the so defined topology on $Q$. Further, we shall assume that for any $\xi\in L$ the map $\phi\mapsto [\xi,\phi]$ is continuous with respect to this topology. Conversely, we assume moreover that any continuous linear map $Q\to\R$ can be obtained as $\phi\mapsto [\xi,\phi]$ for some $\xi\in L$. The latter two conditions may be summarized as saying that the weak topology on the real Hilbert space $Q$ coincides with the topology induced by the maps $\phi\mapsto [\xi,\phi]$. We shall call $\Omega$ \emph{admissible} if and only if these additional assumptions are satisfied.

As is well known, there is no analogue of the Lebesgue measure on infinite-dimensional vector spaces. Thus the Hilbert space of Schrödinger wave functions $\cH^{\sr}$ cannot be defined as a space of square-integrable functions with respect to such a measure. On the other hand, decomposing Schrödinger wave functions $\psi$ by separating a vacuum part as in (\ref{eq:factswv}) lead in Section~\ref{sec:pisom} to an equivalent definition of $\cH^{\sr}$ in terms of a Gaussian measure $\nu_Q$ on $Q$. Recall the definition (\ref{eq:nuq}) of $\nu_Q$ which we rewrite using $\Omega$,
\begin{equation}
\xd\nu_Q(\phi)=\exp(-\Re(\Omega(\phi,\phi)))\,\xd\mu_Q(\phi),
\label{eq:smeasure}
\end{equation}
and where $\mu_Q$ is the Lebesgue measure on $Q$ normalized such that $\nu_Q$ is a probability measure. Now, $\psi\in\cH^{\sr}$ iff $\widetilde{\psi}\in\rL^2(Q,\nu_Q)$, where $\psi=\widetilde{\psi} \cohs_0$. An analogue of the Gaussian measure $\nu_Q$ does not exist either if $Q$ is infinite-dimensional. However, such a measure does exist if we suitably enlarge the space $Q$, see e.g.\ \cite{GeVi:genf4}. Concretely, a suitable extension $\hat{Q}$ of $Q$ may be constructed roughly as follows \cite{Oe:holomorphic}. We consider finite-dimensional quotients $Q_\alpha$ of $Q$, where a suitable Gaussian measure $\nu_\alpha$ is well defined. These spaces turn out to form a projective system of vector spaces with compatible measures. The projective limit $\hat{Q}\defeq\varprojlim Q_\bullet$ then does possess a Gaussian measure, which we shall continue to denote by $\nu_Q$. $\hat{Q}$ is very simple to characterize: It is the algebraic dual space of the topological dual space of $Q$. This characterization also furnishes a canonical inclusion $Q\toi \hat{Q}$.

We define the Hilbert space $\cH^{\st}$ of the Schrödinger representation therefore as the space $\rL^2(\hat{Q},\nu_Q)$ of (equivalence classes of) square-integrable complex functions on $\hat{Q}$ with respect to the measure $\nu_Q$. Note that in this definition the elements of $\cH^{\st}$ are not the wave functions in the usual sense, but their \emph{reduced} versions where the vacuum has been factored out, hence the difference in notation to $\cH^\sr$. Even though these reduced wave functions are functions on $\hat{Q}$, for an important subclass of these it is possible to continue to regard them as functions merely on $Q$. More precisely, a function on $Q$ that is \emph{cylindrical}, i.e., that is induced from a function on a finite-dimensional quotient space $Q_\alpha$, extends canonically to a function on $\hat{Q}$, since $Q_\alpha$ is also a quotient of $\hat{Q}$. In other words, the cylindrical functions are the functions on $Q$ that are \emph{almost translation invariant}, i.e., that are translation invariant with respect to a closed subspace of $Q$ that has finite codimension. The subspace of cylindrical functions in $\rL^2(\hat{Q},\nu_Q)$ is precisely the inductive limit $\varinjlim \rL^2(Q_\bullet,\nu_\bullet)$ of the spaces $\rL^2(Q_\alpha,\nu_\alpha)$.

Since the vacuum wave function (\ref{eq:swvac}) is not cylindrical, it does not extend in a straightforward way to a function on $\hat{Q}$. This means in particular, that we cannot reverse in the infinite-dimensional case the step of going from ``full'' wave functions to reduced wave functions. On the other hand, the reduced wave functions (\ref{eq:swcohi}) of coherent states are cylindrical and thus extend to functions on $\hat{Q}$. These serve thus also in the infinite-dimensional case to define coherent states.
Note that continuous functions on $Q$ that are cylindrical extend to continuous functions on $\hat{Q}$ with respect to its natural topology which is the initial (or projective) topology. In particular, this applies to the coherent state reduced wave functions. Indeed, the topological compatibility conditions we have imposed on $\Omega$ in relation to the other basic structural ingredients of the Schrödinger representation ensure precisely that the reduced wave function $\cohnsr_\tau$ given by (\ref{eq:swcohi}) is continuous on $Q$. Since $\nu_Q$ is a Borel measure on $\hat{Q}$ this implies in particular that $\cohnsr_\tau$ is measurable.
In analogy to the finite-dimensional case we denote the vector subspace of $\cH^{\st}$ generated by the coherent states by $\cH^{\stc}$. Proposition~\ref{prop:scohdensefd} implies that $\cH^{\stc}$ is dense in $\varinjlim \rL^2(Q_\bullet,\nu_\bullet)$ which in turn is dense in $\cH^{\st}$ (compare Proposition~3.1 in \cite{Oe:holomorphic}).

\begin{prop}
\label{prop:scohdenseid}
The subspace $\cH^{\stc}$ is dense in $\cH^{\st}$.
\end{prop}

This completes the definition and basic characterization of the Schrödinger representation in the general case. Let us emphasize again, at this point, the crucial difference between the case of a finite-dimensional and an infinite-dimensional configuration space $Q$. While the Schrödinger representation in the former case is intrinsically defined, its very definition depends in the latter case on an additional datum: The vacuum, encoded in the bilinear form $\Omega$. In this respect the Schrödinger representation becomes then to resemble the holomorphic representation, which both in finite-dimensional and infinite-dimensional cases depends on the complex structure $J$ as additional datum.

\subsection{Relation to the holomorphic representation}
\label{sec:relholom}

The correspondence between complex structures $J:L\to L$ and admissible bilinear forms $\Omega:Q\times Q\to\C$ with positive definite real part extends to the infinite-dimensional case. Moreover, also in the infinite-dimensional case this correspondence is linked to an isomorphism of Hilbert spaces $\bst:\cH^{\st}\to\cH^{\hr}$. We explain this in the present section.

We start with the generalization of Proposition~\ref{prop:fdcorr}.
\begin{prop}
\label{prop:idcorr}
There is a one-to-one correspondence between complex structures $J:L\to L$ and admissible symmetric bilinear forms $\Omega:Q\times Q\to\C$ with positive definite real part.
\end{prop}
The proof of this proposition is substantially the same as that of its counterpart in the finite-dimensional case. The difference is that at various points additional functional analytic considerations enter. We proceed to discuss this in some detail.
Thus suppose we are given a complex structure $J:L\to L$. Now define a symmetric bilinear form $\Omega:Q\times Q\to\C$ as in equation (\ref{eq:jtoomega}). Note that $Q$ inherits a topology from the separable Hilbert space $L$ with inner product $g$ as a quotient space. Indeed, $\Re\Omega$ is precisely the induced inner product on $Q$, making it into a real separable Hilbert space with this topology. Since $[\cdot,\cdot]:L\times L\to\R$ is continuous by assumption on $J$, so is then $[\cdot,\cdot]:L\times Q\to\R$. Moreover, $j:Q\to L$ is continuous so that $\Im\Omega$ is also continuous. Moreover, let $\lambda:Q\to\R$ be linear and continuous. By the Riesz Representation Theorem on $Q$ there exists $\xi\in Q$ such that for all $\phi\in Q$, $\lambda(\phi)=g(j(\phi),j(\xi))=2\omega(j(\phi),J j(\xi))=[-J j(\xi),\phi]$. Thus, all admissibility conditions are satisfied by $\Omega$.

Conversely, suppose we are given an admissible bilinear form $\Omega:Q\times Q\to\C$ with positive definite real part. As in Section~\ref{sec:fdcorr} we define the subspace $X\subseteq L$ according to expression (\ref{eq:defxio}). One of the admissibility conditions is that any continuous linear map $Q\to\R$ can be written as $\phi\mapsto [\xi,\phi]$ for some $\xi\in L$. This can be sharpened to require $\xi\in M$ since $M$ is complement of $N$ in $L$. This then ensures precisely that the argument given there and showing that $X$ is a complement of $M$ in $L$ generalizes to the infinite-dimensional case. The following argument, showing that $X$ is a Lagrangian subspace of $L$ is valid in the infinite-dimensional case without any modification.

Continuing with the steps laid out in Section~\ref{sec:fdcorr}, we define a linear map $\beta:X\to M$ satisfying condition (\ref{eq:defbeta}). Its existence is assured by the same admissibility condition just used previously. The invertibility of $\beta$ follows as in Section~\ref{sec:fdcorr}. As there we proceed to define $J\defeq \beta \oplus \beta^{-1}$. The proof that $J^2(\xi)=-\xi$ and that $\omega(J\xi,J\tau)=\omega(\xi,\tau)$ goes through unchanged in the infinite-dimensional case. This is also true for the proof that the real inner product $g$ induced by $J$ is positive-definite. Also note from the equality of (\ref{eq:ipomega1}) with (\ref{eq:ipomega2}) that $X$ and $M$ are orthogonal subspaces in $L$. Moreover, $q|_X$ and $q\circ\beta^{-1}$ identify $X$ and $M$ with $Q$ as inner product spaces. Since $Q$ is complete and separable this implies that $X$ and $M$ are complete and separable and so is their orthogonal direct sum $L$. To see that $[\cdot,\cdot]:L\times L\to\R$ is continuous we consider its restrictions to $M\times L$ and to $X\times L$. The first restriction coincides with the same restriction of $2\omega$, which is continuous since it may be obtained as a composition of $g$ and $J$, which are continuous by construction. The second restriction coincides due to the definition (\ref{eq:defxio}) with the corresponding restriction of $-\Im\Omega\circ(q\times q)$, which is continuous due to the admissibility conditions.

Finally, the proof that the constructed maps from complex structures to admissible bilinear forms and vice versa are inverse is exactly as in the finite-dimensional case. This completes the proof of Proposition~\ref{prop:idcorr}.

Before proceeding to explain how the relation between the Schrödinger and holomorphic representations extend to the infinite-dimensional case we recall some aspects of the latter \cite{Oe:holomorphic}. In order to define polarized wave functions on $L$ we need a suitable measure on it. For finite-dimensional $L$ this was the Gaussian measure $\nu_L$ given by (\ref{eq:fdhmeasure}). This can be extended to the infinite-dimensional case in the same way as explained for the Schrödinger representation in Section~\ref{sec:infinis}: We consider a projective system of finite-dimensional quotient spaces of $L$ with compatible Gaussian measures and take the projective limit. This projective limit $\hat{L}$ is actually larger than $L$ and may be described as the algebraic dual of the topological dual of $L$. Polarized wave functions are then roughly speaking square-integrable \emph{holomorphic} functions on $\hat{L}$. It turns out, however, that these functions are completely characterized by their values on $L$ alone, rather than on all of $\hat{L}$. (In \cite{Oe:holomorphic} this is Theorem~3.18.) So the resulting Hilbert space $\cH^{\hr}$ is really a space of holomorphic functions on $L$. This is in marked contrast to the Schrödinger representation (Section~\ref{sec:infinis}) where the restriction of wave functions, defined on $\hat{Q}$, to the subspace $Q$ is insufficient to characterize them.

The standard coherent state associated to $\tau\in L$ is represented by the wave function $\cohh_\tau$ in $\cH^{\hr}$ given in (\ref{eq:hcoh}) also in the infinite-dimensional case \cite{Oe:holomorphic}. As in the finite-dimensional case we denote by $\cH^{\hrc}$ the dense subspace of $\cH^{\hr}$ spanned by the coherent states. For both, Schrödinger and the holomorphic representation, the inner product (\ref{eq:ipcoh}) between coherent states takes the same form in the infinite-dimensional case as in the finite-dimensional one. Indeed, for any two concrete coherent states it may be computed using coherent state wave functions on suitable finite-dimensional quotient spaces of $L$ or $Q$.

Given an admissible bilinear form $\Omega$ and a corresponding complex structure $J$ define $\bst:\cH^{\stc}\to\cH^{\hrc}$ as mapping the reduced wave function $\cohnsr_\tau$ determined by (\ref{eq:swcohi}) to the wave function $\cohnh_\tau$ given by (\ref{eq:hncoh}). It is then clear that this map is a bijection preserving the inner product. Since $\cH^{\stc}$ and $\cH^{\hrc}$ are dense subspaces the completion yields an isometric isomorphism $\bst:\cH^{\st}\to\cH^{\hr}$ of Hilbert spaces.
\begin{thm}
\label{thm:liisom}
There is an isometric isomorphism of Hilbert spaces $\bst:\cH^\st\to\cH^\hr$ that maps the reduced Schrödinger wave function $\cohnsr_\tau$ to the holomorphic wave function $\cohnh_\tau$ for all $\tau\in L$.
\end{thm}

We have thus arrived at an infinite-dimensional version of the Segal-Bargmann transform $\bst$ indirectly, i.e., by first working out the finite-dimensional version, applying it to coherent states and generalizing those. In a more direct approach, some of the results of Proposition~\ref{prop:bstfc} generalize, with suitable modifications, to the infinite-dimensional case. However, we shall not detail this approach here.

\subsection{Representation of linear observables}
\label{sec:linobs}

A classical observable in the present setting is a function $F:L\to\R$. Its quantization according to (\ref{eq:gquantobs}) is simplest if the associated Hamiltonian vector field $X_F$ is the complex conjugate of a polarized vector field. Then, given an adapted symplectic potential $\theta$, equations (\ref{eq:adaspot}) and (\ref{eq:polfunc}) ensure that the expression (\ref{eq:gquantobstriv}) simplifies to a multiplication operator,
\begin{equation}
 \check{F} (\psi s) = F \psi s .
\end{equation}

In the Schrödinger polarization, which is real, the condition that $X_F$ be polarized means that at each point $\xi\in L$, the vector $(X_F)_\xi$ be in the subspace $M$ of $L$. But since $M$ is Lagrangian this is equivalent to the derivative of $F$ at $\xi$ vanishing in all directions spanned by $M$. In other words, this is equivalent to $F$ depending only on the quotient space $Q$ of $L$. So the quantization of observables $F:Q\to\R$ that depend only on ``configurations'' is for wave functions $\psi\in\cH^{\sr}$ given by the multiplication rule,
\begin{equation}
 (\check{F} \psi)(\phi) = F(\phi) \psi(\phi) .
\label{eq:squantobsconf}
\end{equation}
Indeed, this is the well known Schrödinger quantization rule for configuration space observables. However, this applies directly only in the case where $Q$ (or equivalently $L$) is finite-dimensional. In the infinite-dimensional case we have to consider instead of $\cH^{\sr}$ the space $\cH^{\st}$ of reduced wave functions. Replacing full wave functions with reduced wave functions has no effect on expression (\ref{eq:squantobsconf}). However, the elements of $\psi\in\cH^{\st}$ are functions on $\hat{Q}$ rather than functions on $Q$. If $Q$ is infinite-dimensional, $\hat{Q}$ is strictly larger than $Q$. So, for (\ref{eq:squantobsconf}) to make sense the observable $F$ needs to be extended to a function on $\hat{Q}$. This is canonically possible in particular, if $F$ is a continuous and almost translation invariant function on $Q$, recall the relevant remarks in Section~\ref{sec:infinis}.

As already mentioned, the formula (\ref{eq:gquantobs}) on its own cannot be used for arbitrary observables, but just the ones that preserve the polarization. In the following we shall limit the discussion to linear continuous observables. These, in particular, preserve the polarization. In that case the Hamiltonian vector field $X_F$ of an observables becomes translation-invariant and we can identify it with an element of $L$ itself. Its defining property (\ref{eq:hvectobs}) may then be written as
\begin{equation}
 2\omega(\xi,X_F)=F(\xi)\qquad\forall\xi\in L .
\label{eq:hamiltobs}
\end{equation}
The non-degeneracy of $\omega$ and (in the infinite-dimensional case) its compatibility with the topology of $L$ ensures that there really is a one-to-one correspondence between continuous linear maps $F:L\to\R$ and elements $X_F\in L$.

In the finite-dimensional case expression (\ref{eq:gquantobstriv}) takes on wave functions $\psi\in\cH^{\sr}$ the form
\begin{align}
 (\check{F} \psi)(\phi) = - [X_F,\phi]\, \psi(\phi) -\im\, (D_{X_F} \psi)(\phi)
\label{eq:squantlinobs}
\end{align}
for $\phi\in Q$. (Note that we have used both (\ref{eq:hamiltobs}) and (\ref{eq:sptosf}) to obtain this.) Here, $D_{\xi} \psi$ denotes the derivative of $\psi$ in the direction $\xi$, i.e., for $\phi\in Q$ and $\xi\in L$,
\begin{equation}
 (D_\xi\psi)(\phi) \defeq \lim_{t\to 0}\frac{\psi(\phi+t q(\xi))-\psi(\phi)}{t} .
\label{eq:sderiv}
\end{equation}

It is interesting to separate the two terms on the right hand side of (\ref{eq:squantlinobs}). The second term is absent for all wave functions iff $X_F\in M$, since precisely then $q(X_F)=0$ which translates to a vanishing of the derivative (\ref{eq:sderiv}). Since $M$ is a Lagrangian subspace of $L$ this is true if and only if $F$ vanishes on $M$, i.e., may be viewed as a function on the quotient space $Q$. In that case we have, moreover, $-[X_F,\xi]=2\omega(\xi,X_F)=F(\xi)$, so we recover the action (\ref{eq:squantobsconf}), as we should.
On the other hand, the first term on the right hand side of (\ref{eq:squantlinobs}) vanishes for all $\phi\in Q$ precisely if $X_F\in N$. Since $N$ is also a Lagrangian subspace of $L$, this is equivalent to $F$ vanishing on $N$. Since $N$ is a complement of the ``space of momenta'' $M$, we might view such a function $F$ as ``depending on momenta only''. Thus, we recover the usual Schrödinger quantization rule that momentum space observables are represented as derivative operators.

Since not all wave functions $\psi\in \cH^{\sr}$ are differentiable (more precisely: have a differentiable representative) the expression (\ref{eq:squantlinobs}) is only well defined on a subspace of $\cH^{\sr}$. Moreover, even if for a wave function $\psi\in\cH^{\sr}$ the expression (\ref{eq:squantlinobs}) can be evaluated, the resulting function $\check{F} \psi$ may not be an element of the Hilbert space $\cH^{\sr}$.
The action of the operator $\check{F}$ given by (\ref{eq:squantlinobs}) is well defined, however, on the dense subspace $\cH^{\src}\subseteq\cH^{\sr}$ spanned by the coherent states. Moreover, it leads there to elements of the Hilbert space $\cH^{\sr}$. Indeed, it is easy to explicitly evaluate the action on a coherent state,
\begin{equation}
 (\check{F} \cohs_\tau)(\phi) = \left( -[X_F,\phi]
 +[\tau,X_F]+\im\Omega(q(X_F),\phi-q(\tau))\right) \cohs_\tau(\phi) .
\label{eq:slocoh}
\end{equation}

Replacing full wave functions with reduced wave functions $\widetilde{\psi}\in\cH^{\st}$, the quantization rule (\ref{eq:squantlinobs}) is easily seen to take the form
\begin{equation}
 (\check{F} \widetilde{\psi})(\phi) = \left(- [X_F,\phi] +\im\Omega(q(X_F),\phi)\right)\,\widetilde{\psi}(\phi)
 -\im\, (D_{X_F} \widetilde{\psi})(\phi) .
\label{eq:srquantlinobs}
\end{equation}
If $Q$ (or $L$) is infinite-dimensional, however, the reduced wave functions are functions on the larger space $\hat{Q}$ rather than on $Q$. In particular, this means that for (\ref{eq:srquantlinobs}) to make sense the function $\phi\mapsto- [X_F,\phi]+\im \Omega(q(X_F),\phi)$ needs to be extended to $\hat{Q}$. Being continuous and linear, the function in question is almost translation invariant in the sense of Section~\ref{sec:infinis} and thus does extend canonically to a continuous linear function $\hat{Q}\to\C$. Observe that in the special case $X_F\in M$, corresponding to $F$ depending on the quotient space $Q$ only, action by multiplication is correctly recovered.

Also in the infinite-dimensional case the action of the quantum observable $\check{F}$ is well defined in particular on the subspace $\cH^{\stc}$ spanned by coherent states. Indeed, the explicit expression can be read off directly from (\ref{eq:slocoh}),
\begin{equation}
 (\check{F} \cohnsr_\tau)(\phi) = \left( -[X_F,\phi]
 +[\tau,X_F]+\im\Omega(q(X_F),\phi-q(\tau))\right) \cohnsr_\tau(\phi) .
\label{eq:silocoh}
\end{equation}
Again, the term in the brackets on the right hand side extends canonically to a continuous linear map $\hat{Q}\to\C$.

In the holomorphic polarization we continue to restrict attention to continuous linear observables $F:L\to\R$. It is useful to decompose $F$ into its holomorphic part $F^+$ and anti-holomorphic part $F^-$. Explicitly,
\begin{equation}
 F(\xi)=F^+(\xi)+F^-(\xi),\quad\text{where}\quad F^{\pm}(\xi)=\frac{1}{2}(F(\xi)\mp \im F(J\xi)) .
\end{equation}
Note that $F^+$ and $F^-$ are necessarily complex valued. We proceed to evaluate (\ref{eq:gquantobstriv}) with respect to the adapted symplectic potential (\ref{eq:kpotlin}) and the polarized section $u$.
For the Hamiltonian vector field $X_{F^+}$ of the holomorphic part $F^+$ equations (\ref{eq:adaspot}) and (\ref{eq:polfunc}) apply. For the Hamiltonian vector field $X_{F^-}$ of the anti-holomorphic part $F^-$ we obtain $\Theta(\xi,X_{F^-})=F^-(\xi)$. Thus, the action of $\check{F}$ on wave functions $\psi\in\cH^{\hr}$ reads,
\begin{align}
 (\check{F} \psi)(\xi) = F^{+}(\xi)\, \psi(\xi) -\im\, (D_{X_{F^-}} \psi)(\xi) .
\label{eq:hquantlinobs}
\end{align}
Again, this is not well defined for all elements of $\cH^{\hr}$. However, it is well defined on the subspace $\cH^{\hrc}$ spanned by the coherent states and leads there to elements in $\cH^{\hr}$. Indeed, the explicit action on a coherent state wave function is,
\begin{align}
 (\check{F} \cohh_\tau)(\xi) & = \left( F^{+}(\xi) + F^{-}(\tau)\right) \cohh_\tau(\xi) \label{eq:hqlocoh} \\
& = \left(F(\xi)+\frac{\im}{2}\{\xi-\tau,X_F\}\right) \cohh_\tau(\xi) \label{eq:hqlocoh1} \\
& = \left(F(\tau)+\frac{\im}{2}\{X_F,\xi-\tau\}\right) \cohh_\tau(\xi) . \label{eq:hqlocoh2} 
\end{align}

In the holomorphic polarization, creation and annihilation operators take a particularly simple form. Indeed, they correspond respectively to the first and the second term on the right hand side of (\ref{eq:hquantlinobs}). That is, a creation operator arises from a holomorphic linear observable and acts by multiplication with the wave function. An annihilation operator arises from an anti-holomorphic linear observable and acts by derivation on the wave function. Moreover, coherent states are eigen-states of annihilation operators, as explicit in (\ref{eq:hqlocoh}). On the other hand, coherent states are often presented as the states that arise from the action of exponentiated creation operators on the vacuum. Indeed, let $\check{F}^+$ be the creation operator corresponding to the linear observables $F$. That is, $\check{F}^+$ is the creation operator part in (\ref{eq:hqlocoh}) or, equivalently, the quantization of $F^+$. Then,
\begin{equation}
 \exp\left(\check{F}^+\right) \cohh_0= \cohh_\tau,
\end{equation}
where $\tau=-J X_F$.

In the finite-dimensional case, the isomorphism $\bst:\cH^{\sr}\to\cH^{\hr}$ is not only an isomorphism of Hilbert spaces, it is for linear observables also an isomorphism of representations. To make this statement precise, we have to remember that not all linear observables are represented on $\cH^{\sr}$ or $\cH^{\hr}$. It is thus appropriate to identify dense subspaces of $\cH^{\sr}$ or $\cH^{\hr}$ where they are represented. Moreover, these dense subspaces should be closed under the action of these observables. It is customary to use to this end the spaces that are generated by the action of (products of) linear observables (or equivalently, creation operators) on the respective vacuum state. We shall denote these by $\cH^{\srp}$ and $\cH^{\hrp}$, respectively. Since we have put much emphasis on coherent states, however, we prefer to consider the larger spaces generated by the action of (products of) linear observables on general coherent states. We shall denote these spaces by $\cH^{\srcp}$ and $\cH^{\hrcp}$, respectively. The relevant statement is then the following.

\begin{prop}
\label{prop:floisom}
$\bst$ restricts to an isomorphism $\cH^{\srcp}\to\cH^{\hrcp}$ of representations of linear observables. That is, let $F:L\to\R$ be linear and $\psi\in \cH^{\srcp}$. Then, $\bst\psi\in\cH^{\hrcp}$ and $\bst \check{F} \psi=\check{F}\bst \psi$.
\end{prop}

The proof of this statement can be performed by grading the spaces $\cH^{\srcp}$ and $\cH^{\hrcp}$ according to the number of creation operators applied and verifying the statement explicitly on basis elements at each degree. We leave the details to the reader.

As in the finite-dimensional case the isomorphism $\bst:\cH^{\st}\to\cH^{\hr}$ becomes an isomorphism of representations of linear continuous observables in the same sense as described above. The only remark in order here is that the space of reduced wave functions $\cH^{\stcp}$ generated by application of linear observables to coherent states consists entirely of continuous almost translation invariant functions. Thus, there is no difficulty in defining these on all of $\hat{Q}$.

\begin{prop}
\label{prop:iloisom}
$\bst$ restricts to an isomorphism $\cH^{\stcp}\to\cH^{\hrcp}$ of representations of continuous linear observables. That is, let $F:L\to\R$ be continuous linear and $\widetilde{\psi}\in \cH^{\stcp}$. Then, $\bst\widetilde{\psi}\in\cH^{\hrcp}$ and $\bst \check{F} \widetilde{\psi}=\check{F}\bst \widetilde{\psi}$.
\end{prop}

We finally recall that geometric quantization for polarization preserving observables is designed to satisfy the correspondence (\ref{eq:crobs}) between the Poisson bracket of classical observables and the commutator of their quantization. In the present setting of continuous linear observables, this correspondence is rigorously satisfied as may be verified explicitly from (\ref{eq:squantlinobs}), (\ref{eq:srquantlinobs}) or (\ref{eq:hquantlinobs}).

\section{Affine field theory}
\label{sec:affine}

\subsection{Basic structures}
\label{sec:abasic}

We proceed in this Section to consider a more general class of field theory: \emph{affine field theory}.
We start by considering the basic structures in analogy to the linear case of Section~\ref{sec:linftstruc}. Firstly, we suppose now that the space of solutions $A$ is an affine space. That is, there exist a corresponding real vector space $L$ with a transitive and free abelian group action $L\times A\to A$, written as addition ``$+$''. This allows to identify canonically all the tangent spaces $T_\eta A$ with $L$. The symplectic potential may then be seen as a map $\theta:A\times L\to\R$, linear in the second argument. We shall switch from here onwards to the notation $\theta(\eta,\xi)$ instead of $\theta_\eta(\xi)$.
Our second key assumption is that the symplectic potential is equivariant with respect to the affine structure in the following sense: There exists a bilinear form $[\cdot,\cdot]:L\times L\to\R$ such that
\begin{equation}
 \theta(\eta+\xi,\tau)=\theta(\eta,\tau)+[\xi,\tau] \qquad \forall \eta\in A,\forall \xi,\tau\in L .
\end{equation}
This implies in turn that the symplectic form is independent of the base point and may be viewed as an anti-symmetric bilinear map $\omega:L\times L\to\R$ given in terms of the bilinear form $[\cdot,\cdot]$ by (\ref{eq:sptosf}). We shall assume moreover that the symplectic form is non-degenerate.

As becomes evident at this point, the structures appearing in linear field theory as discussed in Section~\ref{sec:linftstruc} resurface here in the context of the tangent space $L$. Indeed, we shall import all of these structures with their obvious modified meaning as well as related assumptions to the present affine case, without necessarily enumerating them explicitly. In the few cases where modified definitions are in order we will provide them in this section. We will also provide additional structures relevant only in the affine case.

For consideration of the Schrödinger representation we suppose that $\theta$ is the adapted symplectic potential. We define $M$, $N$ and $Q$ as in Section~\ref{sec:linftstruc}. Then, as in the linear case, $M$ as a subspace of $L$ defines the Schrödinger polarization. Recall that there is a special section $s$ of the prequantum bundle $B$ satisfying (\ref{eq:trivsec}) with respect to $\theta$ as well as (\ref{eq:normsec}). From here onwards we fix this choice of $s$. The sections defining the Schrödinger polarized Hilbert space then take the form $\psi s$, where $\psi$ is a complex function on the quotient space of ``field configurations on the hypersurface''
\begin{equation}
 C\defeq A/M .
\end{equation}
Note that $C$ is an affine space over $Q$. We denote the quotient map $A\to C$ by $c$. We remark that the symplectic potential $\theta$ may be viewed as a map $A\times Q\to\R$.

For the holomorphic representation, i.e., a Kähler polarization, we need as an additional ingredient apart from the classical data already described a \emph{complex structure} on the tangent spaces of $A$. Since these tangent spaces are all canonically identified with $L$ and the symplectic structure is independent of the base point it will suffice to consider a single complex structure $J$ on $L$ as described in Section~\ref{sec:linftstruc}.
In contrast to the linear case a natural choice for the Kähler potential is obtained only via a choice of base point $\eta\in A$,
\begin{equation}
 K^\eta(\zeta)\defeq\frac{1}{2}g(\zeta-\eta,\zeta-\eta) .
\end{equation}
The adapted symplectic potential $\Theta^\eta:A\times L\to\C$ from (\ref{eq:ksympot}) is then,
\begin{equation}
 \Theta^\eta(\zeta,\xi)=-\frac{\im}{2}\{\zeta-\eta,\xi\} .
\end{equation}
Define the complex function $\alpha^\eta$ on $A$ by
\begin{equation}
\alpha^\eta(\zeta)\defeq\exp\left(\frac{\im}{2} \theta(\eta,\zeta-\eta) + \frac{\im}{2} \theta(\zeta,\zeta-\eta)-\frac{1}{4} g(\zeta-\eta,\zeta-\eta)\right) .
\label{eq:alphagqa}
\end{equation}
As is easily verified this satisfies
\begin{equation}
 \xd\alpha^\eta=-\im\alpha^\eta(\Theta^\eta -\theta) .
\end{equation}
Thus, it follows from (\ref{eq:trivid}) that the section $u^\eta\defeq \alpha^\eta s$ of $B$ satisfies (\ref{eq:trivsec}) with respect to $\Theta^\eta$.

\subsection{Pairing and isomorphism in finite dimensions}
\label{sec:pisoma}

We start by considering the Schrödinger representation in the case that $C$ (or equivalently $A$) is finite-dimensional. In the language of Section~\ref{sec:shgquant} the polarized wave function $\psi$ multiplying the section $s$ is a function of the quotient space $C$ rather than of all of $A$. For the inner product (\ref{eq:trivip}) to be normalizable the integration has to be restricted to $C$. $\mu$ is then a Lebesgue measure on $C$ that we shall denote by $\mu_C$. It turns out that it is conveniently normalized by the condition,
\begin{equation}
\int \exp\left(-g(j(\varphi-\eta),j(\varphi-\eta))\right)\,\xd\mu_C(\varphi) =1 .
\end{equation}
Here $\eta\in A$ is an arbitrary base point. Due to the translation invariance of $\mu_C$, the normalization condition is the same for any choice of $\eta$. In the following we shall denote by $\cH^{\sr}$ the complex Hilbert space $\rL^2(C,\mu_C)$ of complex square-integrable functions on $C$ with respect to the measure $\mu_C$. We denote its inner product by $\langle\cdot,\cdot\rangle^{\sr}$.

The holomorphic representation for affine field theory is described in detail in \cite{Oe:affine}. We recall here some key elements. Firstly, since the section $u^\eta$ depends on the choice of a base point $\eta$ it is convenient to use instead the same section $s$ as in the Schrödinger representation as reference. This means that wave functions can be written in the form
\begin{equation}
\psi(\zeta)=\psi^{\hr,\eta}(\zeta-\eta)\alpha^\eta(\zeta),
\label{eq:decahwv}
\end{equation}
where $\psi^{\hr,\eta}$ is a holomorphic and square-integrable function on $L$ with respect to the measure $\nu_L$ given by (\ref{eq:fdhmeasure}) in the finite-dimensional case. Crucially, while the explicit decomposition (\ref{eq:decahwv}) depends on a base point $\eta$, the characterization of $\psi^{\hr,\eta}$ does not. That is, if $\psi^{\hr,\eta}$ for some fixed $\eta\in A$ is a holomorphic function in $\cL^2(L,\nu_L)$ then so is $\psi^{\hr,{\eta'}}$ for any other $\eta'\in A$.\footnote{Here and in the following we shall use the notation $\cL^2$ for the complete positive semi-definite space of square integrable functions. In contrast, $\rL^2$ denotes as before the Hilbert space of equivalence classes of these modulo functions with support of measure zero.} Moreover, the inner product is preserved under a change of base point. We thus obtain a well defined Hilbert space of wave functions which we denote by $\cH^{\ha}$.

By fixing a base-point $\eta\in A$ the coherent states introduced in Section~\ref{sec:pisom} can be readily imported into the affine setting. Concretely, the wave function for such a coherent state corresponding to the solution $\tau\in L$ is obtained by replacing $\psi^{\hr,\eta}$ in (\ref{eq:decahwv}) with $\cohh_\tau$ given by (\ref{eq:hcoh}). We denote the coherent state thus obtained by $K^\eta_\tau$. It is more convenient, however, to use a definition of coherent state intrinsic to the affine setting. Following \cite{Oe:affine} we define the \emph{affine coherent state} $\coha_\eta$ associated to $\zeta\in A$ by
\begin{equation}
\cohah_\eta(\zeta)\defeq \exp\left(\frac{\im}{2}\theta(\eta,\zeta-\eta)
 +\frac{\im}{2}\theta(\zeta,\zeta-\eta)-\frac{1}{4}g(\zeta-\eta,\zeta-\eta)\right) .
\label{eq:cohah}
\end{equation}
This definition is characterized by the reproducing property in $\cH^{\ha}$,
\begin{equation}
\langle\cohah_{\zeta},\psi\rangle^{\ha}=\psi(\zeta) .
\label{eq:coharp}
\end{equation}
Note also that the affine coherent states are already normalized. Indeed, $\coha_\zeta$ coincides, up to a phase, with the normalized version of the coherent state $K^\eta_\tau$ if $\zeta=\eta+\tau$. We denote the subspace of $\cH^{\ha}$ spanned by the coherent states by $\cH^{\hac}$.

Similarly to the linear case we define a continuous function $\bsfa:A\times C\to\C$,
\begin{equation}
\overline{\bsfa(\zeta,\varphi)}\defeq \exp\left(\im\, \theta(\zeta,\varphi-c(\zeta)) -\frac{1}{2}\Omega(\varphi-c(\zeta),\varphi-c(\zeta)) \right) .
\end{equation}
We use this to define an isometric isomorphism of $\cH^{\hac}$ to the subspace $\cH^{\src}$ of $\cH^{\sr}$. As in the linear case, this is precisely the subspace spanned by coherent states in the Schrödinger representation.
\begin{prop}
\label{prop:bstfca}
There is a subspace $\cH^{\src}\subseteq \cH^{\sr}$ and an isometric isomorphism $\bsta:\cH^{\src}\to \cH^{\hac}$ with the following properties:
\begin{align}
 (\bsta\psi)(\zeta) & = \int  \psi(\varphi) \bsfa(\zeta,\varphi)\,\xd\mu_C(\varphi) && \forall \psi\in\cH^{\src} , \label{eq:deffa} \\
 (\bsta^{-1}\eta)(\varphi) & = \int \eta(\zeta) \overline{\bsfa(\zeta,\varphi)}\,\xd\mu_A(\zeta) &&
\forall \eta\in \cH^{\hrc}, \label{eq:deffinva} \\
 \langle \eta, \bsta\psi \rangle^{\ha} & = \int \overline{\eta(\zeta)}\, \psi(c(\zeta))\, \xd\mu_A(\zeta) && \forall \eta\in \cH^{\hac}\, \forall \psi\in\cH^{\src}, \label{eq:ppropfa} \\
 \langle \psi, \bsta^{-1}\eta \rangle^{\sr} & = \int \overline{\psi(c(\zeta))}\, \eta(\zeta)\,\xd\mu_A(\zeta) && \forall \psi\in \cH^{\src}\, \forall \eta\in\cH^{\hac} . \label{eq:ppropfinva}
\end{align}
Also, the holomorphic coherent state wave function $\cohah_\zeta$ associated with $\zeta\in A$ is mapped under $\bsta^{-1}$ to the following wave function:
\begin{equation}
 \cohas_{\zeta}(\varphi)\defeq
 (\bsta^{-1} \cohah_{\zeta})(\varphi)  = \overline{\bsfa(\zeta,\varphi)} .
\label{eq:defcohsa}
\end{equation}
\end{prop}
The proof, very similar to that of Proposition~\ref{prop:bstfc}, is left to the reader. Proposition~\ref{prop:scohdensefd} also generalizes to the affine case.
\begin{prop}
\label{prop:scohdensefda}
The subspace $\cH^{\src}$ is dense in $\cH^{\sr}$.
\end{prop}
The proof, essentially reducible to that of Proposition~\ref{prop:scohdensefd} by choice of a base point, is also left to the reader. Thus, $\bsta$ extends canonically to an isometric isomorphism $\cH^{\sr}\to \cH^{\ha}$ that we shall also denote by $\bsta$.

The emerging scenario is rather similar to that of the linear setting. On the one hand, there is a single Schrödinger representation naturally isomorphic to all holomorphic representations for different complex structures $J$. On the other hand distinct choices of complex structure $J$ induce distinct choices of families of coherent states. As in the linear case these are naturally parametrized by symmetric bilinear forms $\Omega$ with positive definite real part. Proposition~\ref{prop:fdcorr} continues to describe this correspondence. In crucial difference to the linear case, however, there is no distinguished coherent state analogous to the vacuum state. This means, there does not arise a canonical factorization analogous to expression (\ref{eq:factswv}).

\subsection{Schrödinger representation in infinite dimensions}
\label{sec:infinisa}

Before proceeding to construct the Schrödinger representation in the case when the space of solutions $A$ is infinite-dimensional, we recall relevant aspects of the holomorphic representation in that case.

As in the finite-dimensional case, given a base point $\eta\in A$, we decompose wave functions $\psi$ according to equation (\ref{eq:decahwv}). We then require $\psi^{\hr,\eta}$ to be the restriction to $L$ of a holomorphic function in $\cL^2(\hat{L},\nu_L)$. This yields a Hilbert space of wave functions which we call $\cH^{\ha}$. The independence of this Hilbert space structure under choice of base point is not only true in the finite-dimensional, but also in the infinite-dimensional case (Lemma~4.1 in \cite{Oe:affine}). We also remark specifically that the affine coherent state $\coha_\zeta$ associated to $\zeta\in A$ is represented by the wave function $\cohah_\zeta$ given by (\ref{eq:cohah}) also in the infinite-dimensional case.

Moving on to the Schrödinger representation
we suppose we are given an admissible symmetric bilinear form $\Omega:Q\times Q\to\C$ as defined in Section~\ref{sec:infinis}.  As in that section, to obtain a Hilbert space of square-integrable functions, it is necessary to use reduced wave functions. In contrast to the linear setting there is no canonical way to do this. Instead, any choice of base point $\eta\in A$ induces a canonical factorization,
\begin{equation}
\psi(\varphi)=\psi^{\sr,\eta}(\varphi-c(\eta))\cohas_\eta(\varphi) ,
\label{eq:factswva}
\end{equation}
generalizing the factorization (\ref{eq:factswv}). We shall refer to $\psi^{\sr,\eta}$ in (\ref{eq:factswva}) as the $\eta$-reduced wave function of the state $\psi$. Note that $\psi^{\sr,\eta}$ is a function on $Q$ rather than on $C$. We may then proceed similarly to the linear case by declaring the space of $\eta$-reduced wave functions to be the Hilbert space $\rL^2(\hat{Q},\nu_Q)$ if $Q$ is infinite-dimensional. It turns out that choosing any other base point $\eta'\in A$ leads to a canonically equivalent notion of Hilbert space of $\eta'$-reduced wave functions.

\begin{lem}
Let $\eta,\eta'\in A$. If $\psi^{\sr,\eta}\in \cL^2(\hat{Q},\nu_Q)$, then
\begin{equation}
\psi^{\sr,\eta'}(\phi)\defeq \psi^{\sr,\eta}(\phi+q(\eta'-\eta)) \frac{\cohas_\eta(\phi+c(\eta'))}{\cohas_{\eta'}(\phi+c(\eta'))}
\end{equation}
is well defined as a function on $\hat{Q}$ and is element of $\cL^2(\hat{Q},\nu_Q)$. Moreover, the map $\cL^2(\hat{Q},\nu_Q)\to\cL^2(\hat{Q},\nu_Q)$ defined in this way is an isometric isomorphism.
\end{lem}
\begin{proof}
Consider the quotient of coherent states,
\begin{multline}
\beta_{\eta,\eta'}(\phi)\defeq\frac{\cohas_\eta(\phi+c(\eta'))}{\cohas_{\eta'}(\phi+c(\eta'))}
 = \exp\bigg(\im\theta(\eta,\eta'-\eta)-\frac{1}{2}\Omega(q(\eta'-\eta),q(\eta'-\eta))\\
-\im [\eta'-\eta,\phi]-\Omega(q(\eta'-\eta),\phi)\bigg) .
\label{eq:cohsared}
\end{multline}
Crucially, the continuous function $\beta_{\eta,\eta'}:Q\to\C$ is almost translation invariant and hence extends canonically to a continuous function $\hat{Q}\to\C$ which we also denote by $\beta_{\eta,\eta'}$. This makes $\psi^{\sr,\eta'}$ well defined as a function on $\hat{Q}$. To see that the inner product is conserved when changing the base point consider two states $\psi,\chi$. Then,
\begin{align}
\langle\chi^{\sr,\eta'},\psi^{\sr,\eta'}\rangle^{\st}
& = \int_{\hat{Q}} \psi^{\sr,\eta'}(\phi) \overline{\chi^{\sr,\eta'}(\phi)}\,\xd\nu_{Q}(\phi)\\
& = \int_{\hat{Q}} \psi^{\sr,\eta}(\phi+q(\eta'-\eta)) \overline{\chi^{\sr,\eta}(\phi+q(\eta'-\eta))} \nonumber\\
&\qquad \exp\left(-\Re\Omega(2\phi+q(\eta'-\eta),q(\eta'-\eta))\right)
\xd\nu_{Q}(\phi) \label{eq:swvaequiv2}\\
& =\int_{\hat{Q}} \psi^{\sr,\eta}(\phi) \overline{\chi^{\sr,\eta}(\phi)}\,\xd\nu_{Q}(\phi) \label{eq:swvaequiv3}\\
& = \langle\chi^{\sr,\eta},\psi^{\sr,\eta}\rangle^{\st} .
\end{align}
Here we have used Proposition~3.11 of \cite{Oe:holomorphic} in the step from expression (\ref{eq:swvaequiv2}) to expression (\ref{eq:swvaequiv3}).
\end{proof}

As in the linear case we denote the Hilbert space of the infinite-dimensional Schrödinger representation by $\cH^{\st}$. However, we do this in the present case with the understanding that the concrete presentation of a state in $\cH^{\st}$ as an element of $\rL^2(\hat{Q},\nu_Q)$ requires the choice of a base point $\eta\in A$ and is in the form of an $\eta$-reduced wave function. Note that from (\ref{eq:factswva}) we can read off the $\eta'$-reduced wave function $\coha^{\sr,\eta'}_\eta$ of the affine coherent state $\coha_\eta$ associated with $\eta\in A$ as coinciding with expression (\ref{eq:cohsared}), $\coha^{\sr,\eta'}_\eta=\beta_{\eta,\eta'}$. We denote the subspace of coherent states by $\cH^{\stc}$. This is a dense subspace of $\cH^{\st}$. The proof is essentially reduced to that of Proposition~\ref{prop:scohdenseid} by choosing a base point.

\begin{prop}
The subspace $\cH^{\stc}$ is dense in $\cH^{\st}$.
\end{prop}

\subsection{Relation to the holomorphic representation}
\label{sec:relholoma}

As we have seen, the quantization of an affine space of solutions remains determined by a complex structure $J:L\to L$ in the holomorphic case and an admissible bilinear form $\Omega:Q\times Q\to\C$ in the Schrödinger case. As for linear theories the correspondence between the two given by Proposition~\ref{prop:idcorr} translates to an isometric isomorphism $\bsta:\cH^{\st}\to\cH^{\ha}$ also for affine theories, generalizing the finite-dimensional case of Proposition~\ref{prop:bstfca}. To see this it is sufficient to observe that the inner product of coherent states takes the same form in the infinite-dimensional case as in the finite-dimensional case.

\begin{thm}
\label{thm:bsta}
There is an isometric isomorphism of Hilbert spaces $\bsta:\cH^\st\to\cH^\ha$. Moreover, given a base point $\eta\in A$, $\bsta$ maps the $\eta$-reduced Schrödinger wave function $\coha^{\sr,\eta}_\tau$ to the holomorphic wave function $\cohah_\tau$ for all $\tau\in L$.
\end{thm}

\subsection{Representation of affine observables}
\label{sec:aobs}

We consider in this section the representation of observables on the Hilbert spaces obtained in the quantization of affine field theory. A classical observable is thus a function $F:A\to\R$ that we shall require to be continuous.

We start by considering the Schrödinger representation. The first type of observable of interest is one where $F$ only depends on the quotient space $C$, corresponding to the fact that the associated Hamiltonian vector field $X_F$ is polarized, i.e., $(X_F)_\eta\in M$ for all $\eta\in A$. The operator $\check{F}$ that represents the quantization of $F$ is then given by multiplication according to (\ref{eq:gquantobs}). Specifically, for wave functions $\psi$, the action of $\check{F}$ is given by
\begin{equation}
 (\check{F} \psi)(\varphi) = F(\varphi) \psi(\varphi) .
\label{eq:saquantobsconf}
\end{equation}
This is true both, if $\psi\in\cH^{\sr}$ is a full wave function (and $C$ is finite-dimensional) or if $\psi\in\cH^{\st}$ is an $\eta$-reduced wave function for any $\eta\in A$. In the latter case and when $A$ is infinite-dimensional, however, the formula (\ref{eq:saquantobsconf}) has to be treated with care. Namely, $F$ really needs to be a function on the larger space $\hat{C}$ rather than on $C$. Recall that also that if $F$ is an almost translation invariant function on $C$ it canonically extends to $\hat{C}$. Compare similar remarks in Section~\ref{sec:linobs}.

The other type of observable which we shall be interested in, and to which we limit our attention in the following, is a continuous \emph{affine observable}. That is, the derivative $\xd F$ of $F$ is translation-invariant. In other words, there exists a continuous linear map $F_\lino:L\to\R$ such that
\begin{equation}
 F(\eta+\xi)=F(\eta)+F_\lino(\xi)\qquad\forall\eta\in A,\forall\xi\in L .
\end{equation}
The Hamiltonian vector field $X_F$ associated to $F$ is then translation-invariant and we can identify it with an element of $L$. Its defining property may then be written as
\begin{equation}
 2\omega(\xi,X_F)=F_\lino(\xi)\qquad\forall\xi\in L ,
\end{equation}
in analogy to (\ref{eq:hamiltobs}).

It turns out to be convenient to consider affine observables first in the holomorphic representation, before switching back to the Schrödinger representation. Since the wave functions in $\cH^{\ha}$ are trivialized with respect to the section of the prequantum bundle adapted to the symplectic potential $\theta$, suitably rewriting expression (\ref{eq:gquantobstriv}) yields the quantization rule
\begin{align}
 (\check{F} \psi)(\zeta) = \left(F(\zeta)- \theta(\zeta,X_F)\right)\, \psi(\zeta) -\im\, (D_{X_F} \psi)(\zeta) .
\label{eq:hquantlinobsa}
\end{align}
for $\psi\in\cH^{\ha}$ and $\zeta\in A$. This action is well defined in particular on the dense subspace of coherent states $\cH^{\hac}$ and leads there to elements of $\cH^{\ha}$. Concretely, the action on an affine coherent state $\coha_\eta$ for $\eta\in A$ takes the form,
\begin{align}
 (\check{F} \cohah_\eta)(\zeta)
 & = \left(F(\zeta)+\frac{\im}{2}\{\zeta-\eta,X_F\}\right) \cohah_\eta(\zeta)\\
 & = \left(F(\eta)+\frac{\im}{2}\{X_F,\zeta-\eta\}\right) \cohah_\eta(\zeta) .
\end{align}
This directly generalizes (\ref{eq:hqlocoh1}) and (\ref{eq:hqlocoh2}).

We return to the Schrödinger representation and consider first the special case of finite-dimensional $A$. Then, we can consider full wave functions $\psi\in\cH^{\sr}$. As these are trivialized also with respect to the same symplectic potential $\theta$, expression (\ref{eq:hquantlinobsa}) also applies. However, for wave functions $\psi$ that only depend on $C$ rather than on $A$ it is not manifest in this expression that the resulting wave function $\check{F} \psi$ also depends on $C$ only. This can be made explicit at the cost of choosing an arbitrary base point $\eta\in A$. Then,
\begin{align}
 (\check{F} \psi)(\varphi) = \left(F(\eta)- \theta(\eta,X_F)-[X_F,\varphi-c(\eta)]\right)\, \psi(\varphi) -\im\, (D_{X_F} \psi)(\varphi) .
\label{eq:squantlinobsa}
\end{align}
for $\varphi\in C$. In analogy to the linear case, the last term on the right hand side of (\ref{eq:squantlinobsa}) is absent precisely if  $X_F\in M$, which is equivalent to $F$ depending on $C$ only rather than on all of $A$. Then the other terms simplify to the multiplication formula (\ref{eq:saquantobsconf}). On the other hand, if $X_F\in N$, then the multiplicative part of the operator $\check{F}$ becomes a constant.
On the subspace of coherent states $\cH^{\src}$, the action of 
$\check{F}$ is well defined and leads to elements in $\cH^{\sr}$. Explicitly, for the affine coherent state associated to $\zeta\in A$ we obtain,
\begin{equation}
 (\check{F} \cohas_\zeta)(\varphi) = \left(F(\zeta)-[X_F,\varphi-c(\zeta)]+\im\Omega(q(X_F),\varphi-c(\zeta))\right) \cohas_\zeta(\varphi) .
\label{eq:salocoh}
\end{equation}

Moving on to the infinite-dimensional case we have to replace ordinary wave functions by $\eta$-reduced wave functions. For observables that only depend on $C$, the multiplication rule (\ref{eq:saquantobsconf}) remains unchanged. However, it is then essential that the observable be defined on (or canonically extensible to) $\hat{C}$. On the other hand, continuous affine observables $C\to\R$ are always canonically extensible to continuous affine maps $\hat{C}\to\R$. As is easily deduced from the above formulas, the action of $\check{F}$ on an $\eta$-reduced wave function takes the form,
\begin{equation}
 (\check{F} \psi^{\sr,\eta})(\phi) = \left(F(\eta) - [X_F,\phi] +\im\Omega(q(X_F),\phi)\right)\psi^{\sr,\eta}(\phi)
 -\im\, (D_{X_F} \psi^{\sr,\eta})(\phi) .
\label{eq:sarquantlinobs}
\end{equation}
For coherent states, the transition to $\eta$-reduced wave functions is particularly simple and leads to essentially the same expression (\ref{eq:salocoh}). In particular, $\check{F}$ is well defined on the space of coherent states $\cH^{\stc}$ and leads to states in $\cH^{\st}$.

Similarly to the case for linear theories, the isomorphism $\bsta:\cH^{\st}\to\cH^{\ha}$ of Theorem~\ref{thm:bsta} becomes an isomorphism of representations of affine continuous observables. To this end, we consider the subspaces $\cH^{\stcp}$ and $\cH^{\hacp}$ of $\cH^{\st}$ and $\cH^{\ha}$ respectively, generated by applying a finite number of quantized affine observables to $\cH^{\stc}$ and $\cH^{\hac}$ respectively. The statement concerning the isomorphism of representations is then the following.

\begin{prop}
$\bsta$ restricts to an isomorphism $\cH^{\stcp}\to\cH^{\hacp}$ of representations of continuous affine observables. That is, with $F:L\to\R$ continuous and affine $\bsta \check{F}=\check{F}\bsta$ on $\cH^{\stcp}$.
\end{prop}
The proof of this statement is analogous to the proof in the linear case.

Finally, the commutation relations (\ref{eq:crobs}) are rigorously valid for continuous affine observables $F,G$ as considered here.

\section{Further issues}

\subsection{Other real polarizations}
\label{sec:realpol}

Although we have framed our discussion entirely in terms of the Schrödinger polarization, it is really applicable to any real polarization. Indeed, the Schrödinger polarization is obtained if we chose the adapted symplectic potential to be (\ref{eq:sympot}). However, we may make any other convenient choice for the symplectic potential. We shall limit our discussion to the linear case here. (The affine case merely requires additional choices of base points etc.)

Let $L$ be a vector space with a non-degenerate symplectic form $\omega:L\times L\to\R$ and a suitable topology. Given a closed isotropic subspace $M$ of $L$ and a closed isotropic complement $N$ define $[\cdot,\cdot]:L\times L\to\R$ as
\begin{equation}
[m+n,m'+n']\defeq 2\omega(m,n')\quad\forall m,m'\in M,\forall n,n'\in N .
\end{equation}
It is then easy to see that equation (\ref{eq:sptosf}) as well as the relevant further assumptions of Section~\ref{sec:linftstruc} are satisfied. In particular, $M$ is Lagrangian and plays the role of the subspace defining the real polarization.

Thus, the correspondence results (Proposition~\ref{prop:bstfc}, Theorem~\ref{thm:liisom}, Proposition~\ref{prop:floisom}, Proposition~\ref{prop:iloisom}) are true with the Schrödinger polarization replaced by any real polarization with choice of complement. Also, it is thus clear that there is a one-to-one correspondence of vacua between any two real polarizations with choices of complements. Similarly, in the affine case.

\subsection{Minimizing ingredients}
\label{sec:minimizing}

We have parametrized coherent states in the linear case by elements of $L$, which is natural from the point of view of the holomorphic representation due to the reproducing property. However, there are other ways to parametrize them which may appear more natural from the point of view of the Schrödinger representation.

Consider the continuous map $i_N:N\to Q$ given by the restriction of $q$ to the subspace $N$. Since $N$ is a complement of $M$ this is an isomorphism of vector spaces. By the Open Mapping Theorem, $i_N$ is open and thus a (generally non-isometric) isomorphism of Hilbert spaces. Let $Q^{\star}$ be the dual Hilbert space of $Q$ and consider the continuous map $\lambda_M:M\to Q^{\star}$ given by $(\lambda_M(\tau))(\xi)=[\tau,\xi]$. It is clear that $\lambda_M$ is continuous and a vector space isomorphism. Thus, again by the Open Mapping Theorem it is a (generally non-isometric) isomorphism of Hilbert spaces. On the other hand $N\times M\to L$ given by $(n,m)\mapsto n+m$ is an isomorphism of Hilbert spaces (again non-isometric). Thus, we have an isomorphism between $L$ and $Q\times Q^{\star}$. In particular, we may use it to parametrize coherent states by elements of $Q\times Q^{\star}$. Concretely, the reduced Schrödinger wave function (\ref{eq:swcohi}) of the coherent state associated to $(\sigma,\lambda)\in Q\times Q^{\star}$ is given by
\begin{equation}
 \cohnsr_{(\sigma,\lambda)}(\phi) =\exp\left(\Omega(\sigma,\phi)+\im \lambda(\phi)
  -\frac{1}{2}\Omega(\sigma,\sigma)-\frac{\im}{2}\lambda(\sigma)\right) .
 \label{eq:swcohredr}
\end{equation}
In particular, it is possible to construct the Hilbert space $\cH^{\st}$ and its coherent states explicitly with only the ingredients $Q$ and $\Omega$.

In the affine case similar remarks apply and we may parametrize affine coherent states (\ref{eq:defcohsa}) by elements of $C\times Q^{\star}$. Thus, the Schrödinger wave function of the affine coherent state associated to $(\sigma,\lambda)\in C\times Q^{\star}$ is,
\begin{equation}
\cohas_{(\sigma,\lambda)}(\varphi)=\exp\left(\im\, \lambda(\varphi-\sigma) -\frac{1}{2}\Omega(\varphi-\sigma,\varphi-\sigma) \right) .
\end{equation}

\subsection{Comparison with previous results}
\label{sec:CCQ}

The explicit correspondence between complex structures and Schrödinger vacuum wave functions was investigated previously in the context of a specific class of linear theories in \cite{CCQ:schroefock,CCQ:schroecurv}. More precisely, the context was that of linear real scalar field theories in globally hyperbolic curved spacetime. In the following we shall outline how the present framework reproduces, clarifies and generalizes key results of those papers.

As a first step we recast aspects of the presentation of the complex structure and of the Schrödinger vacuum wave function in a form more suitable for a comparison. Our starting point is a vector space $L$ with the additional structures as introduced in Section~\ref{sec:linftstruc}. Recall in particular, that $L$ is the direct sum of its closed subspaces $M$ and $N$. Denote the projectors onto the subspaces $M$ and $N$ corresponding to this decomposition by $p_M$ and $p_N$ respectively. That is, for $m\in M$ and $n\in N$ we define $p_M(m+n)\defeq m$ and $p_N(m+n)\defeq n$. Note that these are orthogonal projectors only if $N=J M$. Nevertheless, they are continuous since $p_N=i_N^{-1}\circ q$ with $i_N$ as defined in Section~\ref{sec:minimizing}. Similarly, for $p_M$.
Define now the linear maps
\begin{align}
& A: N\to N, && A(\xi)\defeq p_N(J\xi),\\
& B: M\to N, && B(\xi)\defeq p_N(J\xi),\\
& C: M\to M, && C(\xi)\defeq p_M(J\xi),\\
& D: N\to M, && D(\xi)\defeq p_M(J\xi) .
\end{align}
The maps $A,B,C,D$ obviously carry the full information about the complex structure $J$. Being composites of continuous maps these are continuous. Moreover, the maps $B$ and $D$ are continuously invertible. To see this for $B$, it is enough to realize that $B$ is the composition of the continuously invertible maps $J|_M: M\to JM$, $q|_{J M}:J M\to Q$ (proof of continuous invertibility as above) and $i_N^{-1}:Q\to N$. The case of $D$ is analogous.
By using the map $i_N$ to identify $Q$ with $N$ we may view the symmetric bilinear form $\Omega:Q\times Q\to\C$ as a symmetric bilinear form $\tilde{\Omega}:N\times N\to\C$. Using (\ref{eq:jtoomega}) it is then straightforward to express $\tilde{\Omega}$ in terms of the maps $A,B,C,D$,
\begin{equation}
\tilde{\Omega}(\phi,\phi')=\Omega(i_N(\phi),i_N(\phi'))=
[B^{-1}(\phi),\phi']-\im [C (B^{-1}(\phi)),\phi'] .
\label{eq:vacabcd}
\end{equation}

We proceed to recall aspects of the treatment presented in \cite{CCQ:schroefock}. The authors consider a linear scalar field theory in a globally hyperbolic spacetime. Fixing a Cauchy hypersurface coordinatized by $x\in\R^3$ their space $L$ of solutions is presented in terms of initial value data on that hypersurface. Concretely, the subspace $N$ of $L$ is given by field configurations $\phi:\R^3\to\R$. The subspace $M$ of $L$ is given by the normal derivatives of field values parametrized by $\pi:\R^3\to\R$ given as,
\begin{equation}
 \pi(x)\defeq \sqrt{h(x)} n^a(x)(\partial_a\phi)(x).
\end{equation}
Here $h$ is the induced Riemannian metric on the hypersurface, $n(x)$ the normal vector at $x$ and $\partial_a$ a coordinate derivative. The symplectic potential as a map $[\cdot,\cdot]:M\times N\to\R$ is then,
\begin{equation}
 [\pi,\phi]=\int \pi(x)\phi(x)\,\xd^3 x .
\label{eq:sympotabcd}
\end{equation}
Of course, elements of $M$ and $N$ really should be equivalence classes of square-integrable functions in a suitable sense, so the above definitions have to be suitably adapted. However, these details are irrelevant for our purposes here and we shall ignore them.
A complex structure $J$ on $L$ is encoded in terms of the operators $A,B,C,D$ as defined above. This is equation (33) in \cite{CCQ:schroefock}. (The minus sign in that equation arises from differing conventions compared to ours.)

Schrödinger wave functions are then viewed as functions on the ``configuration space'' $N$. From our perspective this amounts to canonically identifying $N$ with $Q$ via $i_N$. The vacuum wave function (\ref{eq:swvac}) viewed thus is,
\begin{equation}
\phi\mapsto\exp\left(-\frac{1}{2}\tilde{\Omega}(\phi,\phi)\right) .
\end{equation}
With $\tilde{\Omega}$ given by (\ref{eq:vacabcd}) and the symplectic potential given by (\ref{eq:sympotabcd}) this is,
\begin{equation}
\phi\mapsto\exp\left(-\frac{1}{2}\int \phi(x) ((B^{-1}-\im C B^{-1}) \phi)(x)\,\xd^3 x\right),
\end{equation}
reproducing exactly (65) of \cite{CCQ:schroefock}. The Gaussian measure (\ref{eq:smeasure}) as a measure on $N$ is thus,
\begin{equation}
\xd\nu_N(\phi)=\exp\left(-\Re\tilde{\Omega}(\phi,\phi)\right)\xd\mu_N(\phi)
 =\exp\left(-\int \phi(x) (B^{-1}\phi)(x)\,\xd^3 x\right)\xd\mu_N(\phi),
\end{equation}
where $\mu_N$ is a translation-invariant measure. This yields (42) of \cite{CCQ:schroefock}.

In the infinite dimensional case, wave functions on $N$ have to be replaced by reduced wave functions on a larger space $\hat{N}$, which is called the ``quantum configuration space'' in \cite{CCQ:schroefock}, but is left unspecified there. From our perspective, it is clear that the natural choice for $\hat{N}$ is that of the projective limit of finite-dimensional quotients (recall Section~\ref{sec:infinis}), or equivalently, the algebraic dual of the topological dual of $N$.

The following types of linear observables are considered in \cite{CCQ:schroefock},
\begin{align}
 (\phi,\pi)\mapsto\phi[f]\defeq \int f(x)\phi(x)\,\xd^3 x, \label{eq:cobs} \\
 (\phi,\pi)\mapsto \pi[g]\defeq \int \pi(x) g(x)\,\xd^3 x , \label{eq:pobs}
\end{align}
where $f$ and $g$ are parametrizing real valued functions on the hypersurface (identified with $\R^3$).
Concerning the first type of observable, given by (\ref{eq:cobs}), we may conveniently view the parametrizing function $f$ as an element of $M$ due to (\ref{eq:sympotabcd}). Then, the associated Hamiltonian vector field which we shall call $X_f$ may be written as $X_f=-f$. The quantization (\ref{eq:squantlinobs}) of this observable, which we shall denote by $\check{\phi}[f]$, is simply given by multiplication,
\begin{equation}
(\check{\phi}[f]\psi)(\phi)=\phi[f]\psi(\phi) ,
\label{eq:cact}
\end{equation}
coinciding with equation (25) of \cite{CCQ:schroefock}. For the second type of observable, (\ref{eq:pobs}), we may identify the parametrizing function $g$ with an element of $N$. The associated Hamiltonian vector field, which we shall denote as $X_g$ is thus $X_g=g$. The quantization (\ref{eq:squantlinobs}) of this observable, which we shall denote by $\check{\pi}[g]$ is given by,
\begin{equation}
(\check{\pi}[g]\psi)(\phi)=-\im (D_g \psi)(\phi)=-\im \int g(x) \frac{\delta}{\delta\phi(x)}\,\xd^3 x\; \psi(\phi),
\label{eq:pact}
\end{equation}
where the right-hand expression is a heuristic way to write the derivative. This is (26) in \cite{CCQ:schroefock} (without the ``multiplicative term'').

The above quantization rules (\ref{eq:cact}) and (\ref{eq:pact}) are as expected for the Schrödinger representation. While the transition from full wave functions to reduced wave functions that is necessary in the infinite-dimensional case leaves the first quantization rule invariant, it does modify the second one in a obvious way, recall (\ref{eq:srquantlinobs}). Translated into the current notation this becomes,
\begin{multline}
(\check{\pi}[g]\psi)(\phi)=\im\tilde{\Omega}(g,\phi)\psi(\phi)-\im (D_g \psi)(\phi)\\
=-\im \int \left(g(x)\frac{\delta}{\delta\phi(x)}-\phi(x)((B^{-1}-\im CB^{-1})g)(x)\right)\,\xd^3 x\; \psi(\phi).
\end{multline}
This reproduces the corresponding result (44) of \cite{CCQ:schroefock} and gives it a straightforward explanation.

\section{Discussion and Outlook}
\label{sec:outlook}

While traditionally somewhat neglected in quantum field theory, the Schrödinger representation possesses certain attractive features. One, as noted by Jackiw \cite{Jac:schroedinger}, is the fact that certain aspects of a quantum field theory such as the representation of transformation groups, can be described in a way that is independent of the choice of vacuum. Another attractive feature is the compatibility with the Feynman path integral. In particular, transition amplitudes are obtained from the Feynman path integral when evaluated with states in the Schrödinger representation at the boundaries.

On the other hand, the holomorphic representation is analytically better behaved than the Schrödinger representation. (Recall from Sections~\ref{sec:relholom} and \ref{sec:infinisa} that even in the infinite-dimensional case no extension of the domain on which wave functions are defined is required.) It is also better adapted to the formalism of creation and annihilation operators (as exhibited in Sections~\ref{sec:linobs} and \ref{sec:aobs}). In order to be able to take advantage of the benefits of both representations it is desirable to be able to switch between them. Precisely this should be facilitated by the results presented here.

Linear field theories are the simplest ones that one may consider. Nevertheless, they form the basis of the S-matrix and perturbative treatments. Also, even for linear theories their quantization in curved spacetime is still an issue that may not be viewed as generally resolved. Thus, the presented construction of their Schrödinger representation is hoped to be of some intrinsic interest. Affine field theories present a first generalization of linear field theories. A natural source of affine theories is provided by linear theories on which non-trivial boundary conditions are imposed. Also, connections naturally form affine spaces. However, other aspects of gauge theories such as gauge transformations and a lack of translation-invariance of the symplectic form (in non-abelian gauge theories) are outside of the scope of the present paper.
Our generalization for affine field theories also exhibits certain new features. In particular, the absence of a preferred vacuum state is shown not to be an obstruction to the construction of the Schrödinger representation. The latter may thus also serve as an initial pointer for further generalizations towards more non-linear classes of theories.

Coming back to the two special features of the Schrödinger representation mentioned above, these are both exploited in the Schrödinger-Feynman quantization scheme for the general boundary formulation of quantum theory \cite{Oe:boundary,Oe:gbqft,Oe:kgtl}. In particular, field propagators, which are essentially integral kernels of transition amplitudes, are (heuristically) independent of a choice of vacuum. However, this quantization scheme, while quite successful in certain contexts (see for example \cite{Oe:timelike,Oe:kgtl,CoOe:spsmatrix,CoOe:smatrixgbf,CoOe:2deucl}) has been limited in others due to a lack of rigor. The present work should help improve this situation. For example, for linear and affine field theories rigorous quantization schemes targeting the general boundary formulation have been devised using the holomorphic representation \cite{Oe:holomorphic,Oe:affine}. These are in part (more or less) secretly based on the Feynman path integral. With the present work it will finally be possible to make this connection completely explicit.

\subsection*{Acknowledgments}

I would like to thank Daniele Colosi for many stimulating discussions that have influenced content and presentation, as well as for carefully reading draft versions and alerting me to several minor mistakes in those.

\appendix

\section{Proof of Proposition~\ref{prop:scohdensefd}}
\label{sec:proof}

We need the following adapted version of the Stone-Weierstrass Theorem. Its proof is similar to versions of the Stone-Weierstrass Theorem found in text books such as \cite{Lan:rfanalysis}. Here $\contvi$ denotes the algebra of continuous functions that vanish at infinity. $\contb$ denotes the algebra of continuous functions that are bounded.

\begin{thm}
\label{thm:stonewei}
Let $S$ be a locally compact Hausdorff space and $w\in\contvi(S,\R)$ such that $w(x)>0$ for all $x\in S$. Let $A\subseteq\contb(S,\C)$ be a subalgebra that separates points, vanishes nowhere and is closed under complex conjugation. Then, $w\cdot A$ is dense in $\contvi(S,\C)$ with the topology of uniform convergence.
\end{thm}

\begin{proof}[Proof of Proposition~\ref{prop:scohdensefd}]
We first show that every element of $\rL^2(Q,\nu_Q)$ can be approximated arbitrarily by continuous functions with compact support. Since the subspace of $\rL^2(Q,\nu_Q)$ spanned by characteristic functions of measurable sets is dense, it is sufficient to approximate such a characteristic function. Say we consider the measurable set $X$ and fix $\epsilon>0$. Since $\nu_Q$ is a regular Borel measure, there is a compact set $K$ and an open set $U$ such that $K\subseteq X\subseteq U$ and $\nu_Q(K)+\epsilon > \nu_Q(X) > \nu_Q(U) - \epsilon$. By Urysohn's Lemma there is a continuous function $f:Q\to [0,1]$ such that $f(x)=1$ if $x\in K$ and $f(x)=0$ if $x\notin U$. Since $K$ is compact there is $r>0$ such that $K\subset B_r(0)$. Define the continuous function $g:Q\to [0,1]$ as $g(x)\defeq f(x)$ if $x\in B_r(0)$, $g(x)\defeq 0$ if $x\notin B_{r+1}(0)$ and $g(x)\defeq (r+1-\|x\|) f(x)$ if $x\in B_{r+1}(0)\setminus B_{r}(0)$. Then, $\chi_K(x)\le g(x)\le \chi_U(x)$ for all $x\in Q$. (We follow here the widespread custom to denote the characteristic function of a set $Z$ by $\chi_Z$.) So $(\chi_X-g)(x)=0$ if $x\in K\cup (Q\setminus U)$ while $|(\chi_X-g)(x)|\le 1$ if $x\in U\setminus K$. But $\nu_Q(U\setminus K)< 2\epsilon$. So $\|\chi_X-g\|_2< \sqrt{2\epsilon}$, where $g$ is continuous and has compact support.

Define a measure on $Q$ denoted $\nu'_Q$ and given by
\begin{equation}
\xd\nu'_Q\defeq \frac{1}{w^2} \xd\nu_Q .
\end{equation}
Here $w:Q\to \R^+$ is the continuous function given by
\begin{equation}
 w(\phi)\defeq 2^{n/4} \exp\left(-\frac{1}{4}g(j(\phi),j(\phi))\right) ,
\end{equation}
where $n$ is the dimension of $Q$. $\nu'_Q$ is easily verified to be a probability measure. Note also that $w$ is strictly positive and vanishes at infinity.

Define $A$ to be the vector subspace of $\rL^2(Q,\nu_Q)$ spanned by the functions $\cohnsr_{\tau}$ given by expression (\ref{eq:swvcoh}) for $\tau\in M$. It is easy to verify that $A$ is an algebra of bounded continuous functions that separates points, vanishes nowhere and is closed under complex conjugation.

Suppose finally that $\psi\in\cH^{\sr}$. We shall show that $\psi$ can be arbitrarily approximated in $\rL^2(Q,\mu_Q)$ by elements of $A\cdot \cohs_0$, which is the subspace of $\cH^{\src}$ generated by the coherent states associated to elements in the ``momentum'' subspace $M$ of the space $L$ of solutions. Using the decomposition (\ref{eq:factswv}) this is equivalent to showing that $\tilde{\psi}$ can be arbitrarily approximated in $\rL^2(Q,\nu_Q)$ by elements of $A$. In view of the first part of the proof we might assume without loss of generality that $\widetilde{\psi}$ is continuous and has compact support. We note that $\widetilde{\psi}\cdot w$ is thus also continuous and has compact support, so in particular, vanishes at infinity. Let $\epsilon>0$. By Theorem~\ref{thm:stonewei} there exists $a\in A$ such that for all $\phi\in Q$,
\begin{equation}
 |(\widetilde{\psi}\cdot w- a\cdot w)(\phi)|<\epsilon .
\end{equation}
This implies,
\begin{equation}
\|\widetilde{\psi}-a\|_{\nu_Q,2}^2 =
\int |\widetilde{\psi}-a|^2 \,\xd\nu_Q =
\int |\widetilde{\psi}\cdot w- a\cdot w|^2\,\xd\nu'_Q <\epsilon^2 .
\end{equation}
In particular, $\|\widetilde{\psi}-a\|_{\nu_Q,2}<\epsilon$. This completes the proof.
\end{proof}

\section{Relation to the Segal-Bargmann transform}
\label{sec:bst}

In the context of Section~\ref{sec:pisom} consider the special case $N=JM$. Choosing suitable coordinates, the integral transform $\bst$ given by (\ref{eq:deff}) then reproduces the Segal-Bargmann transform as introduced in \cite{Bar:hilbanalytic}, compare expression (2.3). We detail this in the following.

Choose an orthonormal basis $\{\phi_k\}_{k\in\{1,\dots,n\}}$ of the real Hilbert space $Q$. This gives rise to coordinate functions $q_k:Q\to\R$ defined via
\begin{equation}
 \phi=\sum_{k=1}^n q_k(\phi)\, \phi_k ,\qquad\forall \phi\in Q,
\end{equation}
thus identifying $Q$ with $\R^n$. On the other hand, $\{j(\phi_k)\}_{k\in\{1,\dots,n\}}$ is an orthonormal basis of the real Hilbert subspace $J M$ of $L$. We use this to define coordinate functions $x_k:J M\to \R$ via
\begin{equation}
 \xi=\sqrt{2} \sum_{k=1}^n x_k(\xi)\, j(\phi_k) ,\qquad\forall\xi\in J M,
\end{equation}
identifying $JM$ with $\R^n$. Similarly, $\{J j(\phi_k)\}_{k\in\{1,\dots,n\}}$ is an orthonormal basis of $M$ and we define coordinate functions $y_k:M\to \R$ via,
\begin{equation}
 \xi=\sqrt{2} \sum_{k=1}^n y_k(\xi)\, J j(\phi_k) ,\qquad\forall\xi\in M,
\end{equation}
identifying $M$ with $\R^n$. Combining both via $z_k:L\to\C$ defined by $z_k\defeq x_k+\im y_k$ coordinatizes $L$ as an $n$-dimensional complex vector space in agreement with its complex structure $J$.

The normalization (\ref{eq:muq}) means that in the chosen coordinates $\mu_Q$ becomes the measure $\pi^{-n/2}\xd^n q$ on $\R^n$, where $\xd^n q$ denotes the standard Lebesgue measure. This differs from the conventions in \cite{Bar:hilbanalytic} by the factor $\pi^{-n/2}$. On the other hand, $\nu_L$ determined by (\ref{eq:fdhmeasure}) leads in the chosen coordinates to the measure $\pi^{-n}\exp(-\sum_k |z_k|^2)\, \xd^n x\, \xd^n y$ on $\C^n$. This agrees exactly with the conventions in \cite{Bar:hilbanalytic}.
As is easily verified, the kernel $\bsf$ of the integral transform given by (\ref{eq:sbt}) takes the coordinate form,
\begin{equation}
 B(z,q)=\exp\left(\sqrt{2}\sum_{k=1}^n q_k z_k-\frac{1}{2}\sum_{k=1}^n q_k^2-\frac{1}{2}\sum_{k=1}^n z_k^2\right) .
\end{equation}
Up to the missing normalization factor $\pi^{-n/4}$, this coincides precisely with the corresponding expression (2.1) of \cite{Bar:hilbanalytic}. This relative normalization factor accounts precisely for the different normalization of square-integrable functions on $\R^n$ induced by the relative factor $\pi^{-n/2}$ in the measure, mentioned above. The inverse transform (\ref{eq:deffinv}), which is not everywhere defined, corresponds to expression (2.11) in \cite{Bar:hilbanalytic}.

\bibliographystyle{amsordx} 
\bibliography{stdrefs}
\end{document}